\documentclass[11pt]{article}%

\RequirePackage[OT1]{fontenc}
\RequirePackage{amsthm,amsmath}
\RequirePackage{amssymb}
\RequirePackage[colorlinks,citecolor=blue,urlcolor=blue]{hyperref}

\usepackage[round]{natbib}
\usepackage{multirow}
\usepackage[active]{srcltx}
\usepackage{MnSymbol}
\usepackage{graphicx}
\usepackage{verbatim,enumerate}
\usepackage{rotating, lscape}

\def\R{\mathbb{R}}

\def\bZ{\boldsymbol{Z}}

\numberwithin{equation}{section}
\theoremstyle{plain}

\def\btau{\boldsymbol{\tau}}

\def\bC{\boldsymbol{C}}

\def\bZ{\boldsymbol{Z}}
\def\bz{\boldsymbol{z}}

\def\ss{\boldsymbol{s}}

\newcommand{\U}{U}

\makeatletter
\def\varphiall{\@ifnextchar[{\varphiall@i}{\varphiall@i[]}}
\def\varphiall@i[#1]{\@ifnextchar[{\varphiall@ii{#1}}{\varphiall@ii{#1}[#1]}}
\def\varphiall@ii#1[#2]{\varphi_{\mu,\kappa,\beta_{#1},\sigma^2_{#2}}}
\def\hatvarphiall{\@ifnextchar[{\hatvarphiall@i}{\hatvarphiall@i[]}}
\def\hatvarphiall@i[#1]{\@ifnextchar[{\hatvarphiall@ii{#1}}{\hatvarphiall@ii{#1}[#1]}]}
\def\hatvarphiall@ii#1[#2]{\widehat{\varphi}_{\mu,\kappa,\beta_{#1},\sigma^2_{#2}}}
\makeatother

\def\R{\mathbb{R}}
\def\N{\mathbb{N}}

\def\d{\textrm{d}}

\newtheorem{theo}{Theorem}

\newtheorem{lemma}{Lemma}

\renewcommand{\c}{\boldsymbol{c}}

\newcommand{\norm}[1]{\left\Vert#1\right\Vert}

\DeclareMathOperator*{\argmax}{arg\,max}
\DeclareMathOperator{\var}{Var}

\begin{document}


\baselineskip=28pt \vskip 5mm
\begin{center} {\LARGE{\bf Estimation and prediction
of Gaussian processes
 using generalized Cauchy  covariance model under fixed domain asymptotics}}
\end{center}

\begin{center}\large
Moreno Bevilacqua \footnote{ \baselineskip=10pt
Universidad de Valparaiso, Department of Statistics, 2360102 Valparaiso, Chile.\\
Millennium Nucleus Center
for the Discovery of Structures
in Complex Data , Chile.\\
E-mail:  moreno.bevilacqua@uv.cl \\
},
Tarik Faouzi \footnote{ \baselineskip=10pt
Department of Statistics, Applied Mathematics Research Group, University of BioBio, 4081112  Concepci\'on, Chile\\
E-mail: tfaouzi@ubiobio.cl
}


\end{center}

\baselineskip=17pt \vskip 2mm \centerline{\today} \vskip 2mm

\vspace{2cm}

\begin{abstract}

We study estimation and prediction of Gaussian processes with
covariance model  belonging to the generalized Cauchy (GC) family,
under fixed domain asymptotics.
Gaussian processes with this kind of covariance function
  provide separate characterization of fractal dimension and
  long range dependence,
  an appealing feature in many physical, biological or  geological
  systems.
The results of the paper are classified into three parts.

In the first part, we
characterize the equivalence of two Gaussian measures with GC
covariance functions. Then we provide sufficient conditions for the
equivalence of two Gaussian measures with Mat{\'e}rn (MT) and GC covariance
functions and  two Gaussian measures with Generalized Wendland (GW) and GC covariance
functions.

 In the second part,
we establish strong consistency and asymptotic distribution of the
maximum likelihood estimator of the microergodic parameter associated
to GC covariance model, under fixed domain asymptotics.
The last  part focuses on optimal prediction with GC model
and specifically, we give conditions
 for asymptotic efficiency
prediction and asymptotically correct estimation  of mean square error
using a misspecified GC, MT or GW model.

Our findings are illustrated through a simulation study: the first
compares the finite sample behavior of the maximum likelihood
estimation of the microergodic parameter of the GC model  with the given asymptotic
distribution.  We then compare the finite-sample behavior of the
prediction and its associated mean square error when the true model is GC
and the prediction is performed using the true model
and a  misspecified  GW model.
\end{abstract}


\def\bC{\boldsymbol{C}}
\def\bvarphi{\boldsymbol{\varphi}}
\newpage

\section{Introduction}
Two fundamental steps in   geostatistical analysis are estimating the parameters
of a Gaussian stochastic process  and predicting the process at new locations.
In both steps, the covariance function covers a central aspect.
For instance, mean square error optimal  prediction at an
unobserved site depends on the knowledge of the covariance function.
Since a covariance function must be positive definite, practical
estimation generally requires the selection of some parametric families
of covariances and the corresponding estimation of these parameters.

The maximum likelihood (ML) estimation method is generally considered
the best method for estimating the parameters of covariance models.
Nevertheless, the study of the asymptotics properties of ML estimation,
is complicated by the fact that more than one asymptotic frameworks
can be considered when observing a single realization \citep{Zhang:Zimmerman:2005}.
The increasing domain
asymptotic framework corresponds to the case where the sampling domain increases
with the number of observed data and where the distance between any
two sampling locations is bounded away from 0.
The fixed domain asymptotic framework, sometimes called infill asymptotics \citep{Cressie:1993}, corresponds to the case where more and more data are observed in some
fixed bounded sampling domain.



General results for the asymptotics properties of the ML estimator, under increasing
domain asymptotic framework and some mild regularity conditions, are
given in \cite{Mardia:Marshall:1984} and \cite{Bachoc:2014}. Specifically, they show that ML
estimates are consistent and asymptotically Gaussian with asymptotic
covariance matrix equal to the inverse of the Fisher information
matrix.

Under fixed domain asymptotics, no general results are available for
the asymptotic properties of ML estimation. Yet, some results have
been obtained when assuming the covariance belongs to  Mat{\'e}rn (MT) \citep{Mate:60}  or Generalized Wendland (GW) \citep{Gneiting:2002}  models.
Both families allow for a continuous parameterization of smoothess of the underlying Gaussian process, the GW family being additionally compactly supported
 \citep{Bevilacqua_et_al:2016}.
Specifically, when the smoothness
parameter is known and fixed, not all parameters can be estimated
consistently, when $d=1, 2, 3$, with $d$  the dimension of the Euclidean space. Instead, the ratio of variance and
scale (to the power of a function of  the smoothing parameter), sometimes called
microergodic parameter  is
consistently estimable.
This follows  from   results given in   \cite{Zhang:2004} for the  MT model  and   \cite{Bevilacqua_et_al:2016}  for the GW model.


Asymptotic results for ML estimation of the microergodic parameter of
the MT model can be found in \cite{Zhang:2004},
\cite{Du:Zhang:Mandrekar:2009}, \cite{Wang:Loh:2011}  when the scale parameter is assumed known and fixed.
 \cite{Shaby:Kaufmann:2013}
give strong consistency and asymptotic distribution of the
microergodic parameter when estimating jointly the scale and the
variance parameters and
by means of    a simulation study
they show that the asymptotic  approximation is considerably improved in this case.
Similar results for the microergodic parameter of the GW model can be found in  \cite{Bevilacqua_et_al:2016}.

In terms of prediction, under fixed domain asymptotic,
 \cite{Stein:1988,Stein:1990}
provides conditions under which optimal predictions under a misspecified
covariance function are asymptotically efficient, and mean square
errors converge almost surely to their targets. Stein's conditions translates into the fact
that the true and the misspecified covariances must be compatible,
that is the induced Gaussian measures are equivalent \citep{Sko:ya:1973,Ibragimov-Rozanov:1978}.
A weaker condition, based on ratio  of  spectral densities,
is given in  \cite{Stein:1993}.



In this paper we study ML estimation and prediction of
Gaussian processes, under fixed domain asymptotics, using Generalized Cauchy (GC) covariance model.
GC family of covariance models has been proposed in \cite{Gneiting:Schlather:2004} and deeply studied in \cite{LT2009}.
It is particularly attractive because
 Gaussian processes with such covariance function
allow for any combination
of fractal dimension and Hurst coefficient,
  an appealing feature in many physical, biological or  geological
  systems (see \cite{Gneiting:SS:2012}   and \cite{Gneiting:Schlather:2004} and the references therein).

In particular, we offer the following results.  First, we characterize
the equivalence of two Gaussian measures with covariance functions
belonging to the GC family and sharing the same smoothness parameter.
A consequence of this result is that, as  in MT and GW  covariance models, when the smoothness parameter is known and fixed,
not all parameters can be estimated consistently, under fixed domain
asymptotics.
Then we give sufficient conditions for the equivalence
of two Gaussian measures where the state of truth is represented by a
member of the MT or GC family and the other Gaussian measure has a GC
covariance model.

We then assess the asymptotic properties of the ML estimator of the
microergodic parameter associated with the GC family.  Specifically, for
a fixed smoothness parameter, we establish strong consistency and
asymptotic distribution of the microergodic parameter assuming the
scale parameter fixed and known.  Then, we generalize these
results when jointly estimating with ML the variance and the scale parameter.

Finally, using results in  \cite{Stein:1988} and  \cite{Stein:1993}, we study the
implications of our results on prediction, under fixed domain
asymptotics.
 One remarkable implication is that when the true
covariance belongs to the GC family, asymptotic efficiency
prediction and asymptotically correct estimation of mean square error
can be achieved using a compatible compactly supported GW covariance model.

The remainder of the paper is organized as follows.  In Section 2 we
review some results about MT, GW and GC covariance models.  In
Section 3 we first characterize the equivalence of Gaussian measure
under the GC covariance model.  Then we give  sufficient conditions
for the equivalence of two Gaussian measures with MT and GC and
 two Gaussian measures with
 GW and GC
covariance models.  In Section 4 we establish strong consistency and
asymptotic distribution of the ML estimation of the microergodic
parameter of the GC models, under fixed domain asymptotics.  Section 5
discuss the consequences of our results in terms of prediction,
under fixed domain asymptotics.  Section 6 provides two simulation
studies: the first show how well the given asymptotic distribution of
the microergodic parameter apply to finite sample cases, when
estimating with ML a GC covariance model under fixed domain
asymptotics.  The second compare the finite-sample behavior of the
prediction when using two compatible GC and GW models, when the true model is GC.
The final Section provides a discussion on the
consequence of our results and open problems for future research.

\section{Mat{\'e}rn, Generalized Wendland and Generalized Cauchy covariance models}

This section depicts the main features of the three covariance models involved in the paper.
We denote $\{Z(\ss), \ss \in D \} $ a zero mean Gaussian stochastic process on a
bounded set $D$ of $\R^d$, with stationary covariance function $C:\R^d
\to \R$. We consider the family $\Phi_d$ of continuous mappings
$\phi:[0,\infty) \to \R$ with $0<\phi(0) <\infty$, such that
\begin{equation*} 
{\rm cov} \left ( Z(\ss), Z(\ss^{\prime}) \right )= C(\ss^{\prime}-\ss)=  \phi(\|\ss^{\prime } -\ss \|),
\end{equation*}
with $\ss,\ss^{\prime} \in D$, and $\|\cdot\|$ denoting the Euclidean norm. Gaussian processes with such covariance functions are called weakly stationary and isotropic.

\cite{Shoe38} characterized the family $\Phi_d$ as being scale mixtures
of the characteristic functions of random vectors uniformly
distributed on the spherical shell of $\R^d$, with any probability
measure, $F$:
$$ \phi(r)= \int_{0}^{\infty} \Omega_{d}(r \xi) F(\d \xi), \qquad r \ge 0,$$
with $\Omega_{d}(r)= r^{-(d-2)/2}J_{(d-2)/2}(r)$ and $J_{\nu}$ is  Bessel function of the first kind of order $\nu$.
The family $\Phi_d$ is nested, with the inclusion relation $\Phi_{1} \supset \Phi_2 \supset \ldots \supset \Phi_{\infty}$ being strict, and where $\Phi_{\infty}:= \bigcap_{d \ge 1} \Phi_d$ is the family of mappings $\phi$ whose radial version is positive definite on any $d$-dimensional Euclidean space.

The MT function, defined as:
\begin{equation*}
{\cal M}_{\nu,\alpha,\sigma^2}(r)=
  \sigma^2 \frac{2^{1-\nu}}{\Gamma(\nu)} \left (\frac{r}{\alpha}
  \right )^{\nu} {\cal K}_{\nu} \left (\frac{r}{\alpha} \right ),
  \qquad r \ge 0,
\end{equation*}
is a member of the family
$\Phi_{\infty}$ for any positive values of $\alpha$ and $\nu$. Here,
${\cal K}_{\nu}$ is a modified Bessel function of the second kind of
order $\nu$, $\sigma^2$ is the variance and $\alpha$ a positive
scaling parameter.


We also define $\Phi_{d}^{b}$ as the family that consists of members of $\Phi_d$ being additionally compactly supported on a given interval, $[0,b]$, $b>0$. Clearly, their radial versions are compactly supported over balls of $\R^d$ with radius $b$.
For a given $\kappa>0$, the  GW  correlation function is defined as
\citep{Bevilacqua_et_al:2016,Gneiting:2002}:
\begin{equation} \label{WG2*}
\varphi_{\mu,\kappa,\beta,\sigma^2}(r)= \begin{cases}  \frac{\sigma^2}{B(2\kappa,\mu+1)} \int_{r/\beta}^{1} u(u^2-(r/\beta)^2)^{\kappa-1} (1-u)^{\mu}\,\d u  ,& 0 \leq r/\beta < 1,\\ 0,&r/\beta \geq 1, \end{cases}
\end{equation}
where $B$ denotes the beta function,  $\sigma^2$ is the variance
 and $\beta>0$ is the compact support.
Equivalent representations of  (\ref{WG2*}) in terms of  Gauss hypergeometric function or Legendre polynomials
are given in \cite{Hubb12}.
Closed form solutions of integral~\eqref{WG2*} can be obtained when
$\kappa=k$ with $k\in \N$, the so called  \emph{original Wendland} functions \citep{Wendland:1995}, and, using some results in \cite{Schaback:2011}, when
$\kappa=k+0.5$, the so called \emph{missing Wendland} functions.
Arguments in \cite{Gneiting:2002} and  \cite{Zastavnyi2006} show that, for a given $\kappa>0$, $\varphi_{\mu,\kappa,\beta,\sigma^2} \in \Phi^{\beta}_d$
if and only if
$\mu \ge  (d+1)/2+ \kappa$.
 In this special case $\kappa=0$
the GW  correlation function is defined as:
\begin{equation*}
\varphi_{\mu,0,\beta,\sigma^2}(r)=\left ( 1- r/\beta \right )_{+}^{\mu} =   \begin{cases}   \left ( 1- r \right )^{\mu} ,& 0 \leq r/\beta  < 1,\\ 0,&r/\beta \geq 1, \end{cases}
\end{equation*}
and arguments in \citep{Golubov:1981} show that $\varphi_{\mu,0,\beta,\sigma^2} \in \Phi_d^1$ if and only if $\mu \ge (d+1)/2$.


 The parameters $\nu>0$ and  $\kappa\geq 0$ are crucial for the
differentiability at the origin and, as a consequence, for the degree of the
differentiability of the associated sample paths in the MT and GW models. In particular for a
positive integer $k$, the sample paths  of a Gaussian process are $k$ times differentiable
if and only if  $\nu >k$ in the MT case
and if and only if  $\kappa >k-1/2$ in the GW case.



The smoothness of a Gaussian process can also be described via the
Hausdorff or fractal dimension of a sample path.
The fractal dimension $D \in [d,d+1) $ is a measure of the roughness for
non-differentiable Gaussian processes
and higher values indicating rougher surfaces.
For a given covariance function $\phi \in  \Phi_d$
if
$1-\phi(r)\sim r^{\chi}$ as $r\to 0$
for some $\chi \in (0, 2]$ then the  sample paths of the associated  random process have fractal dimension
$D=d+1-\chi/2$. Here $\chi$ is the so
called fractal index that governs the roughness of sample paths of a
non-differentiable 
Gaussian process.

In the case of a MT model
$\chi=2\nu$ so $D=d+1-\nu$ if $0<\nu<1$ and $d$ otherwise
\citep{Adler:1981,Gneiting:SS:2012}. Thus the MT model permits the full range of
allowable values for the fractal dimension.
In the case of GW family $\chi=2\kappa+1$, so that in this case  $D=d+0.5-\kappa$ if $ 0\leq \kappa<0.5$ and  $d$ otherwise.
Thus  the  GW model does not allow to cover the full range of allowable values for the fractal dimension.

Long-memory dependence
can be defined trough the asymptotic behavior of the covariance function at infinity. Specifically,  for a given covariance function $\phi \in  \Phi_d$, if the power-law
$\phi(r) \sim r^{- \varepsilon +d}$ as $r \to \infty$
holds for some $ \varepsilon \in (0,d]$ the stochastic process is said to have long memory with Hurst coefficient $H=\varepsilon/2$.
MT and GW covariance models do not posses this feature.



A celebrated family of members of $\Phi_{\infty}$ is the GC class   \citep{Gneiting:Schlather:2004}, defined as:

\begin{equation} \label{g-cauchy}
{\cal C}_{\delta,\lambda,\gamma,\sigma^2}(r) = \sigma^2 \left ( 1 + (r/\gamma)^{\delta} \right )^{-\lambda/\delta}, \qquad r \ge 0,
\end{equation}
where the conditions $\delta \in (0,2]$ and $\lambda>0, \gamma >0,\sigma^2>0$ are necessary and sufficient for ${\cal C}_{\delta,\lambda,\gamma,\sigma^2}\in\Phi_{\infty}$.
The parameter $\delta$ is crucial for the
differentiability at the origin and, as a consequence, for the
degree of the differentiability of the associated sample paths.
Specifically,  for $\delta=2$, they are infinitely times differentiable
 and they are not differentiable for $\delta \in (0,2)$.

The GC family represents a breaking point with respect to  earlier literature based on the  assumption of self similarity, since it decouples the fractal dimension and the Hurst effect.
Specifically,
the sample paths of the associated stochastic process have fractal dimension $D = d + 1-\delta/2$ for $\delta \in (0,2)$ and if $\lambda \in (0, d]$ it has long memory with Hurst coefficient $H =  \lambda/2$. Thus, $D$ and $H$ may vary independently of each other  \citep{Gneiting:Schlather:2004,LT2009}.

Fourier transforms of radial versions of members of $\Phi_d$, for a given $d$, have a simple expression, as reported in \cite{Stein:1999} and \cite{Yaglom:1987}. For a member $\phi$ of the family $\Phi_d$, we define its isotropic spectral density as
\begin{equation} \label{FT}
 \widehat{\phi}(z)= \frac{z^{1-d/2}}{(2 \pi)^d} \int_{0}^{\infty} u^{d/2} J_{d/2-1}(uz)  \phi(u) {\rm d} u, \qquad z \ge 0, \end{equation}

and through the paper we use the notation $\widehat{{\cal C}}_{\delta,\lambda,\gamma,\sigma^2}$, $\widehat{{\cal M}}_{\nu,\alpha,\sigma^2}$
and $\widehat{\varphi}_{\mu,\kappa,\beta,\sigma^2}$ for the
spectral density
associated with ${{\cal C}}_{\delta,\lambda,\gamma,\sigma^2}$, ${\cal M}_{\nu,\alpha,\sigma^2}$
and $\varphi_{\mu,\kappa,\beta,\sigma^2}$.
A well-known result about the spectral density of the  Mat{\'e}rn model is the following:
\begin{equation} \label{stein1}
\widehat{{\cal M}}_{\nu,\alpha,\sigma^2}(z)= \frac{\Gamma(\nu+d/2)}{\pi^{d/2} \Gamma(\nu)}
\frac{\sigma^2 \alpha^d}{(1+\alpha^2z^2)^{\nu+d/2}}
, \qquad z \geq 0.
\end{equation}
Define the function  $\mathstrut_1 F_2$  as:
\begin{equation*}
\mathstrut_1 F_2(a;b,c;z)=\sum_{k=0}^{\infty}\frac{(a)_{k}z^{k}}{(b)_{k}(c)_{k}k!}, \qquad z \in\mathbb{R},
\end{equation*}
which is a special case of the generalized hypergeometric functions $\mathstrut_q F_p$ \citep{Abra:Steg:70}, with
$(q)_{k}=  \Gamma(q+k)/\Gamma(q)$ for $k \in  \mathbb{N}\cup\{ 0\}$, being the Pochhammer symbol.
The spectral density of $\varphi_{\mu,\kappa,\beta,\sigma^2}$ for  $\kappa\geq0$  is  given by  \citep{Bevilacqua_et_al:2016}:

$$\widehat{\varphi}_{\mu,\kappa,\beta,\sigma^2}(z)=\sigma^{2}L\beta^{d}\mathstrut_1 F_2\Big(\lambda;\lambda+\frac{\mu}{2},\lambda+\frac{\mu}{2}+\frac{1}{2};-\frac{(z\beta)^{2}}4\Big), \quad  z \geq 0$$
where $\lambda= (d+1)/2+ \kappa$,
and
$L=(\Gamma(2\kappa+\mu+1)\Gamma(\kappa)\Gamma(2\kappa+d))(2^{d}\pi^{\frac{d}{2}}\Gamma(\kappa+\frac{d}{2})\Gamma(\mu+2\lambda)\Gamma(2\kappa))^{-1}$.

For two given   functions $g_1(x)$ and $g_2(x)$, with $g_1(x) \asymp g_2(x)$ we mean
that there exist two constants $c$ and $C$
such that $0<c< C<\infty$ and $c|g_2(x)| \leq |g_1(x)| \leq C|g_2(x)|$ for
each $x$.
The next result follows from \cite{LT2009} and describe    the spectral density  of the GC covariance function and its asymptotic behaviour.

\begin{theo}\label{the3} Let ${\cal C}_{\delta,\lambda,\gamma,\sigma^2}(z) $ be the function defined at Equation~\eqref{g-cauchy}. Then, for  $\gamma>0,\sigma^2>0,\lambda>0$ and $\delta \in (0,2)$ :
\begin{enumerate}
\item \begin{equation*}\label{spectrGW}
\widehat{{\cal C}}_{\delta,\lambda,\gamma,\sigma^2}(z) =-\frac{\sigma^2\gamma^{d/2+1}z^{-d}}{2^{d/2-1}\pi^{d/2+1}}Im\int_0^\infty\frac{{\cal K}_{(d-2)/2}(\gamma t)}{(1+\exp(i\frac{\pi\delta}{2})(t/z)^{\delta})^{\lambda/\delta}}t^{d/2}\text{d}t, \qquad  z \geq0.
\end{equation*}
\item
\begin{equation*}\label{spectrGW2}
\widehat{{\cal C}}_{\delta,\lambda,\gamma,\sigma^2}(z) =\varrho z^{-(d+\delta)}- \mathcal{O}(z^{-(d+2\delta)}) \mbox{ }\mbox{ }\hbox{ for } z\to\infty,
\end{equation*}
\item
\begin{equation*}\label{spectrGW3}
\widehat{{\cal C}}_{\delta,\lambda,\gamma,\sigma^2}(z) \asymp z^{-(d+\delta)}\mbox{ }\mbox{ }\hbox{ for }  z\to\infty,
\end{equation*}
\end{enumerate}
where $\varrho=\frac{2^\delta \sigma^2\lambda\Gamma(\frac{\delta+d}{2})\Gamma(\frac{\delta+2}{2})\sin(\frac{\pi\delta}{2})}{\delta\gamma^{\delta}\pi^{\frac{d}{2}+1}}$.
\end{theo}
The existence of the  spectral density  (\ref{FT})
 is guaranteed
if the integral on the right part  of  (\ref{FT}) is convergent. If the integral does not converge, a generalized covariance function should be considered
and  the spectral density must be defined as   the Fourier transform of a  covariance function in the Schwartz space of test
functions \citep{ojala}. \cite{LT2009}   show that if $\lambda \in (0, d]$  , $i.e.$ under long range depeendence,  $\widehat{{\cal C}}_{\delta,\lambda,\gamma,\sigma^2}(z)$ diverge when $z\to 0^+$.

\section{Equivalence of Gaussian measures with Generalized Cauchy, Mat{\'e}rn and Generalized  Wendland covariance models}\label{sec:EquiMeasure}

Equivalence and orthogonality of probability measures are useful tools when assessing the asymptotic properties of both prediction and estimation for stochastic processes.
Denote with $P_i$, $i=0,1$, two probability measures defined on the same
 measurable space $\{\Omega, \cal F\}$. $P_0$ and $P_1$ are called equivalent (denoted $P_0 \equiv P_1$) if $P_1(A)=1$ for any $A\in \cal F$ implies $P_0(A)=1$ and vice versa. On the other hand,  $P_0$ and $P_1$ are orthogonal (denoted $P_0 \perp P_1$) if there exists an event $A$ such that $P_1(A)=1$ but $P_0(A)=0$. For a stochastic process $\{ Z(\ss), \ss \in \R^d \}$, to define previous concepts, we restrict the event $A$ to the $\sigma$-algebra generated by $\{Z(\ss), \ss\in D\}$ where $D \subset \R^d$. We emphasize this restriction by saying that the
two measures are equivalent on the paths of $\{Z(\ss), \ss\in D\}$.

 Gaussian measures are completely characterized by their mean and covariance function.
We write $P(\rho)$ for a Gaussian measure with zero mean and covariance function $\rho$.  It is well known that two Gaussian measures  are either equivalent or orthogonal on the paths of $\{Z(\ss), \ss\in D\}$ \citep{Ibragimov-Rozanov:1978}.

Let $P(\rho_i)$, $i=0, 1$ be two zero mean Gaussian measures with  isotropic covariance function  $\rho_i$ and  associated spectral density $\widehat{\rho}_i$,  $i=0, 1$, as defined through~\eqref{FT}.
Using results in  \cite{Sko:ya:1973} and  \cite{Ibragimov-Rozanov:1978},  \cite{Stein:2004}
has shown that, if for some $a>0$,
$\widehat{\rho}_0(z)z^a$ is bounded away from 0 and $\infty$ as $z \to \infty$, and
for some finite and positive~$c$,
\begin{equation}\label{spectralfinite2}
\int_{c}^{\infty} z^{d-1} \;\left\{ \frac{\widehat{\rho}_1(z)-\widehat{\rho}_0(z)}{\widehat{\rho}_0(z)} \right\}^2\;
{\rm d} z <\infty,
\end{equation}
then for any bounded subset
$D\subset \R^d$, $P(\rho_0)\equiv P(\rho_1)$ on the paths of $Z(\ss), \ss \in D$ .
For the rest of the paper, we denote with $P({\cal M}_{\nu,\alpha,\sigma^2})$, $P(\varphi_{\mu,\kappa,\beta,\sigma^2})$,
$P({\cal C}_{\delta,\lambda,\gamma,\sigma^2})$
a zero mean Gaussian measure induced by a MT, GW and GC covariance function respectively.
The following Theorem is due to  \cite{Zhang:2004}. It  characterizes the compatibility of two MT covariance models
sharing a common smoothness parameter $\nu$.
\begin{theo} \label{Thm2}
For a given $\nu>0$,
let $P({\cal M}_{\nu,\alpha_i,\sigma^2_i})$,  $i=0, 1$, be two zero mean Gaussian  measures. For any
bounded infinite set $D\subset \R^d$, $d=1, 2, 3$,
$P({\cal M}_{\nu,\alpha_0,\sigma^2_0}) \equiv P( {\cal M}_{\nu,\alpha_1,\sigma^2_1})$ on the paths of $Z(\ss), \ss \in D$, if and only if
\begin{equation}\label{condmat}
\frac{\sigma_0^2}{ \alpha_0^{2\nu}}=\frac{\sigma_1^2}{ \alpha_1^{2\nu}}.
\end{equation}
\end{theo}

The following Theorem is a generalization of Theorem 4 in  \cite{Bevilacqua_et_al:2016} and it characterizes the compatibility of two GW covariance  models
sharing a common smoothness parameter $\kappa$.
We omit the proof since the result can be obtained using the same arguments.

\begin{theo} \label{W_vs_W}
 For a given $\kappa \ge 0$, let $P( \varphi_{\mu_i,\kappa,\beta_i,\sigma^2_i})$,  $i=0, 1$, be two zero mean Gaussian  measures and
 let $\mu_i > d+\kappa+1/2$.
 For any
bounded infinite set $D\subset \R^d$, $d=1, 2, 3$,
$P( \varphi_{\mu_0,\kappa,\beta_0,\sigma^2_0}) \equiv P( \varphi_{\mu_1,\kappa,\beta_1,\sigma^2_1})$ on the paths of $Z(\ss), \ss \in D$ if and only if
\begin{equation} \label{condition1_iff}
\frac{\sigma_0^2}{\beta_0^{2 \kappa+1}}\mu_0 =\frac{\sigma_1^2}{\beta_1^{2 \kappa+1}}\mu_1.
\end{equation} \end{theo}

The first relevant result of this paper concerns the characterization of the compatibility of two GC functions sharing a common smoothness parameter.

\begin{theo} \label{C_vs_C}
For a
given $\delta \in (d/2,2)$, let  $P({\cal C}_{\delta,\lambda_i,\gamma_i,\sigma_i^2})$  $i=0, 1$ be two zero mean Gaussian  measures. For any
bounded infinite set $D\subset \R^d$, $d=1, 2, 3$, $P( {\cal C}_{\delta,\lambda_0,\gamma_0,\sigma_0^2}) \equiv P({\cal C}_{\delta,\lambda_1,\gamma_1,\sigma_1^2})$ on the paths of $Z(\ss), \ss \in D$ if and only if    \begin{equation} \label{condition1_iff}
\frac{\sigma_0^2}{\gamma_0^{\delta}} \lambda_0 = \frac{\sigma_1^2}{\gamma_1^{\delta}} \lambda_1.
\end{equation}

\end{theo}

\begin{proof}

Let us start with the sufficient part of the assertion. From  Theorem~\ref{the3} point 3,
$k<z^{d+\delta}\widehat{C}_{\delta,\lambda_0,\gamma_0,\sigma_0^2}(z)<K$ as  $z\to\infty$.
In order to prove the sufficient part, we need to find conditions
such that for some positive and finite $c$,

\begin{equation}\label{Chau0}
\int_c^{\infty} z^{d-1}\left( \frac{\widehat{C}_{\delta,\lambda_1,\gamma_1,\sigma_1^2}(z)-\widehat{C}_{\delta,\lambda_0,\gamma_0,\sigma_0^2}(z)}{\widehat{C}_{\delta,\lambda_0,\gamma_0,\sigma_0^2}(z)}\right)^2 dz
<\infty
\end{equation}
We proceed by direct construction, and, using Theorem~\ref{the3} Point 2 we find that as $z\to\infty$,

\begin{equation*}
\begin{aligned}
\Big|\frac{\widehat{C}_{\delta,\lambda_1,\gamma_1,\sigma_1^2}(z)-\widehat{C}_{\delta,\lambda_0,\gamma_0,\sigma_0^2}(z)}{\widehat{C}_{\delta,\lambda_0,\gamma_0,\sigma_0^2}(z)}  \Big|&\leq \frac{z^{d+\delta}}{k}\Big|\varrho_1 z^{-(d+\delta)}-\mathcal{O}(z^{-(d+2\delta)})-\varrho_0 z^{-(d+\delta)}+\mathcal{O}(z^{-(d+2\delta)})\Big|\\
&\leq \frac{1}{k}\Big|\varrho_1  -\varrho_0 +\mathcal{O}(z^{-\delta})\Big|
\end{aligned}
\end{equation*}

where $\varrho_{i}=\frac{2^\delta \sigma_i^2\lambda_i\Gamma(\frac{\delta+d}{2})\Gamma(\frac{\delta+2}{2})\sin(\frac{\pi\delta}{2})}{\delta\gamma_i^{\delta}\pi^{\frac{d}{2}+1}}$, with $i=0,1$.

Then we obtain,
\begin{equation*}
\begin{aligned}
\int_c^{\infty} z^{d-1}\left( \frac{\widehat{C}_{\delta,\lambda_1,\gamma_1,\sigma_1^2}(z)-\widehat{C}_{\delta,\lambda_0,\gamma_0,\sigma_0^2}(z)}{\widehat{C}_{\delta,\lambda_0,\gamma_0,\sigma_0^2}(z)}\right)^2 dz
&\leq\frac{z^{d+\delta}}{k^2}\int_c^{\infty}z^{d-1}\left(\varrho_1  -\varrho_0 +\mathcal{O}(z^{-\delta})\right)^2{\rm d}z\\
\end{aligned}
\end{equation*}

We conclude  that~\eqref{Chau0} is true if  $ \delta>d/2$  and  $\varrho_0=\varrho_1$.
This last condition implies (\ref{condition1_iff}). Moreover since  $\delta <2 $, the condition $ \delta>d/2$ can be  satisfied only for $d=1, 2, 3.$
 The sufficient part of our claim is thus proved.
 The necessary   part follows the arguments in the proof of \cite{Zhang:2004}.

\end{proof}

An immediate consequence of Theorem~\ref{C_vs_C} is that, for a fixed  $\delta \in (d/2,2)$, the parameters  $\lambda$, $\gamma$ and $\sigma^2$  cannot be estimated consistently.
Nevertheless the microergodic parameter $\sigma^{2}\lambda /\gamma^{\delta}$
is consistently estimable. In Section 4, we establish the asymptotic properties of ML estimation associated with the microergodic parameter of the GC model.

The second relevant result of this paper give sufficient conditions for the compatibility of  a GC and a  MT covariance model.

\begin{theo}\label{ThmX}
For  given $\delta \in (d/2,2)$, let $P( {\cal C}_{\delta,\lambda_1,\gamma_1,\sigma_1^2})$  and $P({\cal M}_{\nu,\alpha,\sigma^2_0})$   be  two zero mean Gaussian measures.
If   $\nu=\delta/2$
and  
\begin{equation}\label{122}
\frac{\sigma_{0}^{2}}{\alpha^{2\nu}}=\left( \frac{ \Gamma^2(\delta/2)\sin(\pi\delta/2)}{2^{1-\delta}\pi} \right) \frac{\sigma_{1}^{2}}{\gamma_1^{\delta}}\lambda_1,
\end{equation}
then  for any
bounded infinite set $D\subset \R^d$, $d=1, 2, 3$,
 $P({\cal M}_{\nu,\alpha,\sigma^2_0}) \equiv  P( {\cal C}_{\delta,\lambda_1,\gamma_1,\sigma_1^2})$ on the paths of $Z(\ss), \ss \in D$,

\end{theo}
\begin{proof}
The spectral density of the  MT model is given by:
\begin{equation} \label{stein1}
\widehat{{\cal M}}_{\nu,\alpha,\sigma_0^2}(z)= \frac{\Gamma(\nu+d/2)}{\pi^{d/2} \Gamma(\nu)}
\frac{\sigma^2 \alpha^d}{(1+\alpha^2z^2)^{\nu+d/2}}
, \qquad z \ge 0.
\end{equation}
It is known that $\widehat{{\cal M}}_{\nu,\alpha,\sigma^2_0}(z)z^a$ is bounded away from 0 and $\infty$ as $z \to \infty$
for some $a>0$ \citep{Zhang:2004}.
In order to prove the sufficient part
we need to find conditions such that  for some positive and finite $c$,
\begin{equation}\label{eq:999}
\int_{c}^{\infty}z^{d-1} \bigg(
\frac{\widehat{C}_{\delta,\lambda_1,\gamma_1,\sigma_1^2}(z)-\widehat{{\cal M}}_{\nu,\alpha,\sigma^2_0}(z)}{\widehat{{\cal M}}_{\nu,\alpha,\sigma^2_0}(z)}
 \biggr)^{2} {\rm d} z<\infty.
\end{equation}

Let $\varrho^{-1}_2=\frac{ \Gamma(\nu+d/2)\sigma_{0}^{2}\alpha^{-2\nu}}{\pi^{d/2}\Gamma(\nu)   }$.
Using asymptotic expansion of ~\eqref{stein1} and  Theorem~\ref{the3}, point 2, we have  that as $z\to\infty$,
\begin{equation*}
\begin{aligned}
\Big|\frac{\widehat{C}_{\delta,\lambda_1,\gamma_1,\sigma_1^2}(z)-\widehat{{\cal M}}_{\sigma^2_0,\alpha,\nu}(z)}{\widehat{{\cal M}}_{\sigma^2_0,\alpha,\nu}(z)}\Big|&= \Big|\varrho^{-1}_2\big[\varrho_1 z^{-(d+\delta)}-\mathcal{O}(z^{-(d+2\delta)})](\alpha^{-2}+z^{2})^{\nu+\frac{d}{2}}-1\Big|\\
&=\Big|\varrho^{-1}_2\big[\varrho_1 z^{-(d+\delta)}-\mathcal{O}(z^{-(d+2\delta)})]z^{2\nu+d}((\alpha z)^{-2}+1)^{\nu+\frac{d}{2}}-1\Big|\\
&=\Big|\varrho^{-1}_2\big[\varrho_1 z^{-(d+\delta)}-\mathcal{O}(z^{-(d+2\delta)})]z^{2\nu+d} \big[1+(\nu+d/2)(\alpha z)^{-2}\\
&\hbox{ }+\mathcal{O}(z^{-2})\big]-1\Big|\\
&=\Big|\varrho^{-1}_2\varrho_1 z^{2\nu-\delta}-1+\varrho^{-1}_2\varrho_1(\nu+d/2)\alpha^{-2}z^{2\nu-\delta-2}+\mathcal{O}(z^{2\nu-\delta-2})\\
&\hbox{ }-\mathcal{O}(z^{2\nu-2\delta})-\mathcal{O}(z^{2\nu-2\delta-2})\Big|\\
&\leq\Big|\varrho^{-1}_2\varrho_1 z^{2\nu-\delta}-1\Big|+\varrho^{-1}_2\varrho_1(\nu+d/2)\alpha^{-2}z^{2\nu-\delta-2}+
\mathcal{O}(z^{2\nu-2\delta})\\
&\hbox{ }+\mathcal{O}(z^{2\nu-2\delta-2})+\mathcal{O}(z^{2\nu-\delta-2}).\\
\end{aligned}
\end{equation*}

Then, if $2\nu=\delta$ and $\varrho^{-1}_2\varrho_1=1$ we obtain,
\begin{equation*}
\begin{aligned}
\int_c^{\infty}z^{d-1}\Big|\frac{\widehat{C}_{\delta,\lambda_1,\gamma_1,\sigma_1^2}(z)-\widehat{{\cal M}}_{\sigma^2_0,\alpha,\nu}(z)}{\widehat{{\cal M}}_{\sigma^2_0,\alpha,\nu}(z)}\Big|^2{\rm d}z
&\leq\int_c^{\infty}z^{d-1}\left((\nu+d/2)\alpha^{-2}z^{-2}+
\mathcal{O}(z^{-\delta})\right)^2{\rm d}z\\
\end{aligned}
\end{equation*}

and the second term of the inequality is finite for  $\delta>d/2$.
Moreover since  $\delta <2 $, the condition $\delta>d/2$ can be  satisfied only for $d=1, 2, 3.$
Then for a given    $\delta\in (d/2,2)$ and 
 $d=1, 2, 3$, inequality  ~\eqref{eq:999} is true if $\nu=\delta/2$ and $\varrho^{-1}_2\varrho_1=1$. This last two conditions
 implies
(\ref{122}).

\end{proof}

{\bf Remark I}:
As expected, compatibility  between GC and MT covariance models is achieved only for a  subset of the parametric space of $\nu$
that leads to  non differentiable sample paths and in particular for  $d/4<\nu<1$, $d=1,2,3$.

The following are sufficient conditions given in  \cite{Bevilacqua_et_al:2016} concerning the compatibility of  a MT and a  GW covariance models.

\begin{theo}\label{matgw}
For given $\nu\geq1/2$ and $\kappa\geq0$,
let $P({\cal M}_{\nu,\alpha,\sigma_0^2})$ and $P( \varphiall[][1])$
  be two zero mean Gaussian measures. If $\nu=\kappa+1/2$,  $\mu> d+\kappa+1/2$, and
\begin{equation}\label{cafu}
\frac{\sigma_{0}^{2}}{\alpha^{2\nu}}=\mu\left( \frac{\Gamma(2\kappa+\mu+1)}{\Gamma(\mu+1)} \right) \frac{\sigma_{1}^{2}}{\beta^{2\kappa+1}},
\end{equation}
then for any
bounded infinite set $D\subset \R^d$, $d=1, 2, 3$, $P({\cal M}_{\nu,\alpha,\sigma^2_0})
  \equiv P(\varphiall[][1])$ on the paths of $Z(\ss), \ss \in D$.
\end{theo}

Putting together  Theorem~\ref{ThmX} and Theorem~\ref{matgw} we obtain the next  new result that establish sufficient conditions
for the compatibility of  a GW and GC covariance function:

\begin{theo}\label{W_C}
For   given  $\delta \in (d/2,2)\cap [1,2)$
let   $P( {\cal C}_{\delta,\lambda,\gamma,\sigma_0^2})$ and $P(\varphi_{\mu,\kappa,\beta,\sigma_1^2})$
 be  two zero mean Gaussian measures.  If  $\kappa+1/2=\delta/2$,  $\mu> d+\kappa+1/2$ and 
\begin{equation}\label{12}
 \left( \frac{\Gamma(2\kappa+\mu+1)}{\Gamma(\mu+1)} \right) \frac{\sigma_{1}^{2}}{\beta^{2\kappa+1}} \mu= \left( \frac{ \Gamma^2(\delta/2)\sin(\pi\delta/2)}{2^{1-\delta}\pi} \right) \frac{\sigma_{0}^{2}}{\gamma^{\delta}}\lambda,
\end{equation}

then for any
bounded infinite set $D\subset \R^d$, $d=1, 2, 3$, $P({\cal C}_{\delta,\lambda,\gamma,\sigma_0^2}) \equiv  P(\varphi_{\mu,\kappa,\beta,\sigma_1^2})$ on the paths of $Z(\ss), \ss \in D$.

\end{theo}


{\bf Remark II}:
As expected, compatibility  between GC and GW covariance models is achieved only for a  subset of the parametric space of $\kappa$
that leads to  non differentiable sample paths and in particular $0\leq\kappa<1/2$, $d=1,2$ and $1/4\leq\kappa<1/2$, $d=3$.

\section{Asymptotic properties of the ML estimation for the Generalized Cauchy model}

We now focus on the microergodic parameter  $\sigma^2\lambda/ \gamma^{\delta}$ associated with  the GC family. The following results fix the  asymptotic properties of its ML estimator. In particular, we shall show that the microergodic parameter can be estimated consistently, and then assess the asymptotic distribution of the ML estimator.

Let $D\subset \R^d$ be  a bounded subset of $ \R^d$
and  $S_n=\{ \ss_1,\ldots,\ss_n \in D \subset \R^d \}$
 denote any set of distinct locations.
Let $\bZ_n=(Z(\boldsymbol{s}_1),\ldots,Z(\boldsymbol{s}_n))^{\prime}$
be a finite  realization of  $Z(\boldsymbol{s})$, $\boldsymbol{s}\in D$, a zero mean stationary Gaussian process with  a given parametric covariance function
$\sigma^2 \phi(\cdot; \btau)$, with $\sigma^2>0$, $\btau$ a parameter vector and  $\phi$ a member of the family $\Phi_d$, with $\phi(0; \btau)=1$.

We then write 
$R_{n}(\btau)=[\phi(\|\boldsymbol{s}_i-\boldsymbol{s}_j\|; \btau)]_{i,j=1}^n$ for the associated correlation matrix.
The  Gaussian log-likelihood function is defined as:
\begin{equation}\label{eq:17}
\mathcal{L}_{n}(\sigma^{2},\btau)=-\frac{1}{2} \left(n\log(2\pi\sigma^{2})+\log(|R_{n}(\btau)|)+\frac{1}{\sigma^{2}}\bZ_n^{\prime}R_{n}(\btau)^{-1}\bZ_n \right).
\end{equation}
Under the GC model, the Gaussian log-likelihood is obtained with $\phi(\cdot; \btau)\equiv {\cal C}_{1,\lambda,\delta,\gamma}$
and $\btau=(\lambda,\delta,\gamma)^{\prime}$.
Since in what follows  $\delta$  and $\lambda$ are  assumed known and  fixed, for notation convenience, we write $\tau=\gamma$.
Let $\hat{\sigma}^2_n$ and  $\hat{\gamma}_n$ be the maximum likelihood estimator obtained maximizing
$\mathcal{L}_{n}(\sigma^{2},\gamma)$ for  fixed $\delta$  and $\lambda$.

In order to prove
consistency and asymptotic Gaussianity of the microergodic parameter, we first consider an estimator that maximizes~\eqref{eq:17} with respect to $ \sigma^{2}$ for  a fixed arbitrary scale parameter $\gamma>0$,
obtaining the following estimator
 \begin{equation} \label{hoceini}
 \hat{\sigma}_n^2(\gamma)=\argmax_{\sigma^2} \mathcal{L}_{n}(\sigma^{2},\gamma)=\bZ_n^{\prime}R_{n}(\gamma)^{-1}\bZ_n/n. \end{equation}
  Here $R_{n}(\gamma)$ is the  correlation matrix  coming from the GC family ${\cal C}_{1,\lambda,\delta,\gamma}$.
 The following result offers some asymptotic properties of ML estimator of  the migroergodic parameter
 $\hat{\sigma}_{n}^{2}(\gamma)\lambda/\gamma^{2\delta}$ both in terms of consistency and asymptotic distribution.
 The proof is omitted since it follows the same steps  in  \cite{Bevilacqua_et_al:2016} and \cite{Wang:Loh:2011}.

 \begin{theo}\label{theo10}
 Let $Z(\boldsymbol{s})$, $\boldsymbol{s}\in D$, be a zero mean Gaussian process  with covariance function belonging to the GC family, i.e.
 ${\cal C}_{\sigma_0^2,\lambda,\delta,\gamma_0}$, with 
$\delta\in (d/2,2)$, $d=1,2,3$ and $\lambda>d$.
 Suppose $(\sigma_{0}^{2},\gamma_0)\in (0,\infty)\times  (0,\infty)$. For a fixed $\gamma>0$, let $\hat{\sigma}_{n}^{2}(\gamma)$ as defined through Equation~\eqref{hoceini}. Then,  as $n\to\infty$, 
\begin{enumerate}
\item $\hat{\sigma}_{n}^{2}(\gamma)\lambda/\gamma^{\delta}\stackrel{a.s}{\longrightarrow} \sigma_{0}^{2}\lambda/\gamma_{0}^{\delta}$ and
\item $n^{\frac{1}{2}}(\hat{\sigma}_{n}^{2}(\gamma)\lambda/\gamma^{2\delta}-\sigma_{0}^{2}\lambda/\gamma_{0}^{\delta})\stackrel{\mathcal{D}}{\longrightarrow}\mathcal{N}(0,2(\sigma_{0}^{2}\lambda/\gamma_{0}^{\delta})^{2})$.
\end{enumerate}
\end{theo}

The second type of estimation considers the  joint maximization of~\eqref{eq:17}  with respect
to
$(\sigma^{2},\gamma)\in (0,\infty)\times I$ where $I=[\gamma_L, \gamma_U]$ and $0<\gamma_L<\gamma_U<\infty$.
The solution of this optimization problem is  given by
$(\hat{\sigma}_n^2(\hat{\gamma}_n),\hat{\gamma}_n)$ where
 $$\hat{\sigma}_n^2(\hat{\gamma}_n)= \bZ_n^{\prime}R_{n}(\hat{\gamma}_n)^{-1}\bZ_n/n$$
and $\hat{\gamma}_n=\argmax_{\gamma \in I} \mathcal{PL}_{n}(\gamma)$.  Here $\mathcal{PL}_{n}(\gamma)$
is the profile log-likelihood:
\begin{equation}\label{eq:prof}
\mathcal{PL}_{n}(\gamma)=-\frac{1}{2} \left( \log(2\pi)+n\log(\hat{\sigma}_n^2(\gamma))+\log|R_{n}(\gamma)| +n \right).\end{equation}
We now establish the asymptotic properties of the  sequence of random variables  $\hat{\sigma}_{n}^{2}(\hat{\gamma}_{n})\lambda/\hat{\gamma}_{n}^{\delta}$ in a special case.
The following  Lemma is needed in order to establish consistency and asymptotic distribution.





\begin{lemma}\label{lemmaML}
 For any $ \gamma_1< \gamma_2$, $\gamma_i \in I=[\gamma_L, \gamma_U]$, $i=1,2$ and  $\delta\in(0,1]$ and $\lambda >d$ then
 $\hat{\sigma}_{n}^{2}(\gamma_1)/\gamma_1^{ \delta} \leq \hat{\sigma}_{n}^{2}(\gamma_2)/\gamma_2^{ \delta}$ for each $n$.
\end{lemma}

\begin{proof}
The proof follows  \cite{Shaby:Kaufmann:2013} and \cite{Bevilacqua_et_al:2016}  which use the same arguments in the MT and GW cases.
Let  $0<\gamma_1<\gamma_2$, with $ \gamma_1, \gamma_2 \in I$. Then, for any $\bZ_n$,
$$\hat{\sigma}_{n}^{2}(\gamma_1)/\gamma_1^{\delta}- \hat{\sigma}_{n}^{2}(\gamma_2)/\gamma_2^{ \delta}=\frac{1}{n}
\bZ_n^{\prime}(R_{n}(\gamma_1)^{-1}\gamma_1^{-\delta} -R_n(\gamma_2)^{-1}\gamma_2^{-\delta})\bZ_n $$
is nonnegative if the matrix  $R_{n}(\gamma_1)^{-1}\gamma_1^{-\delta} -R_n(\gamma_2)^{-1}\gamma_2^{-\delta}$  is positive semi-definite and this happens if and only if  the matrix
  $B=R_{n}(\gamma_2)\gamma_2^{ \delta} -R_n(\gamma_1)\gamma_1^{ \delta}$
  with generic element
    $$B_{ij}=\gamma_2^{ \delta}{\cal C}_{\delta,\lambda,\gamma_2,1} (\norm{\ss_i-\ss_j}) -  \gamma_1^{ \delta} {\cal C}_{\delta,\lambda,\gamma_1,1}(\norm{\ss_i-\ss_j}). $$
  is  positive semi-definite.
From Theorem {\it 3.3} of \cite{Tar_Mor2018}, this happens  if $\delta\in(0,1]$ and $\lambda>d$.
\end{proof}

We now establish strong consistency and asymptotic distribution
 of the sequence of random variables  $\hat{\sigma}_{n}^{2}(\hat{\gamma}_{n})\lambda/\hat{\gamma}_{n}^{\delta}$.
\begin{theo}\label{theo11}
Let $Z(\boldsymbol{s})$, $\boldsymbol{s}\in D\subset \R^d$,  be a zero mean Gaussian process with a Cauchy covariance model
${\cal C}_{\sigma_0^2,\lambda,\delta,\gamma_0}$ with
  $d=1$ and $\delta\in (1/2,1]$, $\lambda>1$ 
or $d=2$ and $\delta=1$, $\lambda>2$
 Suppose $(\sigma_{0}^{2},\gamma_{0})\in (0,\infty)\times
I$ where  $I=[\gamma_L, \gamma_U]$ with $0<\gamma_L<\gamma_U<\infty$. Let
$(\hat{\sigma}_{n}^{2},\hat{\gamma}_{n})^{\prime}$ maximize~\eqref{eq:17} over $(0,\infty)\times I$. Then as
$n\to\infty$,
\begin{enumerate}
\item $\hat{\sigma}_{n}^{2}(\hat{\gamma}_{n})\lambda/\hat{\gamma}_{n}^{\delta}   \stackrel{a.s}{\longrightarrow}\sigma_0^{2}(\gamma_0)\lambda/\gamma_0^{\delta}$ and
\item $\sqrt{n}(\hat{\sigma}_{n}^{2}(\hat{\gamma}_{n})\lambda/\hat{\gamma}_{n}^{\delta}  -\sigma_0^{2}(\gamma_0)\lambda/\gamma_0^{\delta})\stackrel{\mathcal{D}}{\longrightarrow} \mathcal{N}(0,2(\sigma_0^{2}(\gamma_0)\lambda/\gamma_0^{\delta})^{2})$.
\end{enumerate}
\end{theo}
\begin{proof}
The proof follows  \cite{Shaby:Kaufmann:2013} and \cite{Bevilacqua_et_al:2016}  which use the same arguments in the MT and GW cases.
Let ${\cal G}_n(x)=\hat{\sigma}_{n}^{2}(x)/x^{ \delta}  $ and
 define  the sequences ${\cal G}_n(\gamma_L)$      and ${\cal G}_n(\gamma_U)$.
Since $\gamma_{L}\leq \hat{\gamma}_{n} \leq \gamma_{U}$  for every $n$, then, using Lemma \ref{lemmaML},
${\cal G}_n(\gamma_L) \leq {\cal G}_n(\hat{\gamma}_{n})    \leq {\cal G}_n(\gamma_U)$
for all $n$ with probability
one. Combining this with Theorem \ref{theo10}  implies the result.

\end{proof}




\section{Prediction using Generalized Cauchy model}

We now consider  prediction of a Gaussian process at a new location $\ss_0$,
using GC model, under fixed domain asymptotic.
Specifically, we focus on
two properties: asymptotic efficiency  prediction and asymptotically correct
estimation of prediction variance.
\cite{Stein:1988} shows that both asymptotic
properties hold when the  Gaussian measures are equivalent.
Let $P({\cal C}_{\sigma_i^2,\lambda_i,\delta,\gamma_i})$, $i=0,1$ be two zero mean Gaussian measures. Under
$P({\cal C}_{\sigma_0^2,\lambda_0,\delta,\gamma_0})$, and using Theorem \ref{C_vs_C},  both properties hold when
$\sigma_0^{2}\lambda_0\gamma_0^{-\delta}=\sigma_1^{2}\lambda_1\gamma_1^{-\delta}$, $\delta \in (d/2,2)$ and $d=1,2,3$.

Similarly, let $P( {\cal M}_{\nu,\alpha,\sigma_2^2})$  and $P({\cal C}_{\sigma^2,\lambda,\delta,\gamma})$
be two Gaussian measures with MT and  Cauchy model.
Using Theorem \ref{ThmX},
under
 $P({\cal M}_{\nu,\alpha,\sigma_2^2})$      both properties hold when (\ref{122}) is true,
$\delta \in (d/2,2)$, $d=1,2,3$.
In addition, let $P(\varphi_{\mu,\kappa,\beta,\sigma_3^2})$ a Gaussian measure with GW model.
Using Theorem \ref{W_C},
under
$P(\varphi_{\mu,\kappa,\beta,\sigma_3^2})$     both properties hold when (\ref{12}) is true,  $\mu>d+\kappa+1/2$, $\delta \in (d/2,2)\cap [1,2)$ and   $d=1,2,3$.

 Actually, \cite{Stein:1993} gives a substantially weaker condition, based on the ratio of spectral densities, for asymptotic efficiency  prediction based on
the asymptotic behaviour of the ratio of the isotropic spectral densities. Now, let
\begin{equation}\label{blup}
\widehat{Z}_{n}(\delta,\lambda,\gamma)=\c_n(\delta,\lambda,\gamma)^{\prime}R_{n}(\delta,\lambda,\gamma)^{-1}\bZ_n
\end{equation}
be the best linear unbiased predictor at an unknown location $\ss_0\in D\subset \R^d$,
under the misspecified model ${\cal C}_{\delta,\lambda,\gamma,\sigma^2}$,
where $\c_n(\delta,\lambda,\gamma)=[{\cal C}_{\delta,\lambda,\gamma,1}(\|\boldsymbol{s}_0-\boldsymbol{s}_i)\|]_{i=1}^n$
 and $R_{n}(\delta,\lambda,\gamma)=[{\cal C}_{\delta,\lambda,\gamma,1}(\|\boldsymbol{s}_i-\boldsymbol{s}_j)\|]_{i,j=1}^n$
 is the correlation matrix.

If the correct model is $P( {\cal C}_{\delta,\lambda_0,\gamma_0,\sigma^2_0})$, then the mean squared error of the predictor is given by:
\begin{align}\label{mse_miss}
&\var_{\delta,\lambda_0,\gamma_0,\sigma^2_0}\left[\widehat{Z}_{n}(\delta,\lambda,\gamma)-Z(\boldsymbol{s}_0)\right]=\sigma_0^2\Big(1-2\c_n(\delta,\lambda,\gamma,)^{\prime}R_{n}(\delta,\lambda,\gamma)^{-1}\c_n(\delta,\lambda_0,\gamma_0)\\ &\quad+ \c_n(\delta,\lambda,\gamma)^{\prime}R_{n}(\delta,\lambda,\gamma)^{-1} R_{n}(\delta,\lambda_0,\gamma_0) R_{n}(\delta,\lambda,\gamma)^{-1}\c_n(\delta,\lambda,\gamma)\Big)\nonumber.
\end{align}
If
$\gamma_0= \gamma$ and $\lambda_0=\lambda$, i.e., true and wrong models coincide, this expression simplifies to
\begin{align}\label{msetrue}
\var_{\delta,\lambda_0,\gamma_0,\sigma^2_0}\big[&\widehat{Z}_{n}(\delta,\lambda_0,\gamma_0)-Z(\boldsymbol{s}_0)\big]\\
 \nonumber   &=\sigma_0^2\big(1-\c_n(\delta,\lambda_0,\gamma_0)^{\prime}R_{n}(\delta,\lambda_0,\gamma_0)^{-1}\c_n(\delta,\lambda_0,\gamma_0)\big).
\end{align}
Similarly  $\var_{\nu,\alpha,\sigma^2_2}\big[\widehat{Z}_{n}(\delta,\lambda,\gamma)-Z(\boldsymbol{s}_0)\big]$, $\var_{\nu,\alpha,\sigma^2_2}\big[\widehat{Z}_{n}(\nu,\alpha)-Z(\boldsymbol{s}_0)\big]$ and $\var_{\mu,\kappa,\beta,\sigma^2_3}\big[\widehat{Z}_{n}(\delta,\lambda,\gamma)-Z(\boldsymbol{s}_0)\big]$, $\var_{\mu,\kappa,\beta,\sigma^2_3}\big[\widehat{Z}_{n}(\mu,\kappa,\beta)-Z(\boldsymbol{s}_0)\big]$ can be defined  under, respectively, $P({\cal M}_{\nu,\alpha,\sigma_2^2})$ and $P(\varphi_{\mu,\kappa,\beta,\sigma^2_3})$, where $\widehat{Z}_{n}(\nu,\alpha)$ and $\widehat{Z}_{n}(\mu,\kappa,\beta)$ are the best linear unbiased predictor using respectively the MT and GW models.
The following results are an application of Theorems~1 and~2 of \cite{Stein:1993}.

\begin{theo}\label{kauf_3}
Let  $P({\cal C}_{\delta,\lambda_0,\gamma_0,\sigma_0^2})$, $P({\cal C}_{\delta,\lambda_1,\gamma_1,\sigma_1^2})$, $P(\varphi_{\mu,\kappa,\beta,\sigma^2_3})$,
$P( {\cal M}_{\nu,\alpha,\sigma^2_2})$ be four Gaussian probability measures on $D\subset \R^d$ with
$\delta \in (d/2,2)$ and $d=1,2,3$. Then, for all $\boldsymbol{s}_0\in D$:
\begin{enumerate}
  \item Under $P({\cal C}_{\delta,\lambda_0,\gamma_0,\sigma_0^2})$, as $n\to \infty$,
  \begin{equation}\label{kauf3_1} \frac{\var_{\delta,\lambda_0,\gamma_0,\sigma^2_0}\bigl[\widehat{Z}_{n}(\delta,\lambda_1,\gamma_1)-Z(\boldsymbol{s}_0)\bigr]}{\var_{\delta,\lambda_0,\gamma_0,\sigma^2_0}\bigl[\widehat{Z}_{n}(\delta,\lambda_0,\gamma_0)-Z(\boldsymbol{s}_0)\bigr]}{\,\longrightarrow\,}1,
        \end{equation}
        for any fixed $\gamma_1>0$ and
          if $\sigma_0^{2}\lambda_0\gamma_0^{-\delta}=\sigma_1^{2}\lambda_1\gamma_1^{-\delta}$, then as $n\to \infty$,
 \begin{equation}\label{kauf3_2} \frac{\var_{\delta,\lambda_1,\gamma_1,\sigma^2_1}\bigl[\widehat{Z}_{n}(\delta,\lambda_1,\gamma_1)-Z(\boldsymbol{s}_0)\bigr]}{\var_{\delta,\lambda_0,\gamma_0,\sigma^2_0}\bigl[\widehat{Z}_{n}(\delta,\lambda_1,\gamma_1)-Z(\boldsymbol{s}_0)\bigr]}{\,\longrightarrow\,}1.
 \end{equation}

   \item Under $P( {\cal M}_{\nu,\alpha,\sigma^2_2})$,   if $\nu=\frac{\delta}{2}$  as $n\to \infty$,
  \begin{equation}\label{kauf3_11}
  \frac{\var_{\nu,\alpha,\sigma^2_2}\bigl[\widehat{Z}_{n}(\delta,\lambda_1,\gamma_1)-Z(\boldsymbol{s}_0)\bigr]}{\var_{\nu,\alpha,\sigma^2_2}\bigl[\widehat{Z}_{n}(\nu,\alpha)-Z(\boldsymbol{s}_0)\bigr]}{\,\longrightarrow\,}1,
        \end{equation}
           for any fixed $\gamma_1>0$ and  
            if 
$\left( \pi^{-1}2^{\delta-1}\Gamma^2(\delta/2)\sin(\pi\delta/2) \right) \sigma_{1}^{2}\lambda\gamma_1^{-\delta}=\sigma_{2}^{2}\alpha^{-2\nu}$,
then as $n\to \infty$
 \begin{equation}\label{kauf3_3}\frac{\var_{\delta,\lambda_1,\gamma_1,\sigma^2_1}\bigl[\widehat{Z}_{n}(\delta,\lambda_1,\gamma_1)-Z(\boldsymbol{s}_0)\bigr]}{\var_{\nu,\alpha,\sigma^2_2}\bigl[\widehat{Z}_{n}(\delta,\lambda_1,\gamma_1)-Z(\boldsymbol{s}_0)\bigr]}{\,\longrightarrow\,}1.
 \end{equation}

   \item Under $P(\varphi_{\mu,\kappa,\beta,\sigma^2_3})$,   if $\kappa+1/2=\delta/2$, $\mu> d+\kappa+1/2$ and
   $\delta \in (d/2,2)\cap [1,2)$  as $n\to \infty$,
  \begin{equation}\label{kaka}
   U_1(\beta)= \frac{\var_{\mu,\kappa,\beta,\sigma^2_3}\bigl[\widehat{Z}_{n}(\delta,\lambda_1,\gamma_1)-Z(\boldsymbol{s}_0)\bigr]}{\var_{\mu,\kappa,\beta,\sigma^2_3}\bigl[\widehat{Z}_{n}(\mu,\kappa,\beta)-Z(\boldsymbol{s}_0)\bigr]}{\,\longrightarrow\,}1,
        \end{equation}
           for any fixed $\gamma_1>0$  and
           if
$\left( \Gamma(2\kappa+\mu+1)\Gamma^{-1}(\mu+1) \right) \sigma_{3}^{2}\mu\beta^{-(2\kappa+1)}=\left(  \pi^{-1} 2^{\delta-1}\Gamma^2(\delta/2)\sin(\pi\delta/2) \right) \sigma_{1}^{2}\lambda\gamma_1^{-\delta}$, then as $n\to \infty$
 \begin{equation}\label{kaka2}U_2= \frac{\var_{\delta,\lambda_1,\gamma_1,\sigma^2_1}\bigl[\widehat{Z}_{n}(\delta,\lambda_1,\gamma_1)-Z(\boldsymbol{s}_0)\bigr]}{\var_{\mu,\kappa,\beta,\sigma^2_3}\bigl[\widehat{Z}_{n}(\delta,\lambda_1,\gamma_1)-Z(\boldsymbol{s}_0)\bigr]}{\,\longrightarrow\,}1.
 \end{equation}

\end{enumerate}
\end{theo}

\begin{proof}
Since $\widehat{{\cal C}}_{\sigma^2,\lambda,\delta,\gamma}(z)$ is bounded away from zero and infinity  and  as $z\to \infty$,
\begin{equation*}\label{lim2}
\begin{aligned}
&\frac{\widehat{{\cal C}}_{\delta,\lambda_1,\gamma_1,\sigma_1^2}(z)}{\widehat{{\cal C}}_{\delta,\lambda_0,\gamma_0,\sigma_0^2}(z)}
&=\frac{\varrho_1 z^{-(d+\delta)}- \mathcal{O}(z^{-(d+2\delta)}) }
{\varrho_0 z^{-(d+\delta)}- \mathcal{O}(z^{-(d+2\delta)}) }
\end{aligned}
\end{equation*}
 where $\varrho_i$, $i=0,1$ are defined in the Proof of Theorem \ref{C_vs_C},
 then, if $\delta \in (d/2,2)$ and $d=1,2,3$
 \begin{equation}\label{lim2}
 \underset{z \to \infty}{\lim}\frac{\widehat{{\cal C}}_{\delta,\lambda_1,\gamma_1,\sigma_1^2}(z)}{\widehat{{\cal C}}_{\delta,\lambda_0,\gamma_0,\sigma_0^2}(z)}=\frac{\varrho_1}{\varrho_0}
 =\frac{\sigma_1^{2}\lambda_1\gamma_1^{-\delta}}{\sigma_0^{2}\lambda_0\gamma_0^{-\delta}},
\end{equation}

and using Theorem~1 of \cite{Stein:1993},  we obtain~\eqref{kauf3_1}.
If $\sigma_1^{2}\lambda_0\gamma_1^{-\delta}=\sigma_0^{2}\lambda_1\gamma_0^{-\delta}$, using Theorem~2 of \cite{Stein:1993},  we obtain~\eqref{kauf3_2}.

Similarly, since  $\widehat{{\cal M}}_{\nu,\alpha,\sigma^2_2}(z)$ is bounded away from zero and infinity
and  as $z\to \infty$
\begin{equation}\label{lim}
\begin{aligned}
\frac{\widehat{{\cal C}}_{\delta,\lambda_1,\gamma_1,\sigma_1^2}(z)}{\widehat{{\cal M}}_{\nu,\alpha,\sigma^2_2}(z)}
&=\varrho^{-1}_2\big[\varrho_1 z^{-(d+\delta)}-\mathcal{O}(z^{-(d+2\delta)})](\alpha^{-2}+z^{2})^{\nu+\frac{d}{2}}\\
&=\varrho^{-1}_2\varrho_1 z^{2\nu-\delta}+\varrho^{-1}_2\varrho_1(\nu+d/2)\alpha^{-2}z^{2\nu-\delta-2}+\mathcal{O}(z^{2\nu-\delta-2})
\hbox{ }-\mathcal{O}(z^{2\nu-2\delta})-\mathcal{O}(z^{2\nu-2\delta-2})
\end{aligned}
\end{equation}
 where $\varrho^{-1}_2$ is defined in the Proof of Theorem \ref{ThmX},
then  if $2\nu=\delta$,
$\delta \in (d/2,2)$ and $d=1,2,3$
$$\underset{z\to\infty}{\lim}\frac{\widehat{{\cal C}}_{\delta,\lambda_1,\gamma_1,\sigma_1^2}(z)}{\widehat{{\cal M}}_{\nu,\alpha,\sigma^2_2}(z)}=\varrho^{-1}_2\varrho_1=\frac{  \Gamma^2(\delta/2)\sin(\pi\delta/2) \sigma_{1}^{2}\lambda_1\gamma_1^{-\delta}}{2^{1-\delta}\pi \sigma_{2}^{2}\alpha^{-2\nu}}.$$
Using Theorem~1  of \cite{Stein:1993}, we obtain~\eqref{kauf3_11}.
If
\begin{equation}\label{ppp1}
\left( 2^{\delta-1}\Gamma^2(\delta/2)\sin(\pi\delta/2)\pi^{-1}  \right) \sigma_{1}^{2}\lambda_1\gamma_1^{-\delta}=\sigma_{2}^{2}\alpha^{-2\nu},
\end{equation}
using Theorem~2 of \cite{Stein:1993},  we obtain~\eqref{kauf3_3}.

Similarly, since $\widehat{{\cal \varphi}}_{\mu,\kappa,\beta,\sigma^2_3}(z)$
 is bounded away from zero and infinity,
 if $2\kappa+1=\delta$,  $\mu>d+\delta/2$,
   $\delta \in (d/2,2)\cap [1,2)$,  $d=1,2,3$ and using the asymptotic results on the spectral density of the GW model in \cite{Bevilacqua_et_al:2016},  we have:

\begin{equation*}
\begin{aligned}
  \underset{z\to\infty}{\lim}\frac{\widehat{{\cal C}}_{\delta,\lambda_1,\gamma_1,\sigma_1^2}(z)}{\widehat{{\cal \varphi}}_{\mu,\kappa,\beta,\sigma^2_3}(z)}
&=\underset{z\to\infty}{\lim}\frac{\varrho z^{-(d+\delta)}- \mathcal{O}(z^{-(d+2\delta)})}{\sigma^2L\beta^{d}\Big[c_{3}(z\beta)^{-(d+1)-2\kappa}\big\{1+\mathcal{O}(z^{-2})\big\}+\,c_{4}(z\beta)^{-(\mu+\frac{d+1}{2}+\kappa)}\big\{\cos(z\beta-c_{5})+\mathcal{O}(z^{-1})\big\}\Big]}\\
&=\frac{\varrho}{\sigma^2L\beta^{-(2\kappa+1)}c_{3}}=\frac{  2^{\delta-1}\Gamma^2(\delta/2)\sin(\pi\delta/2)\pi^{-1}  \sigma_{1}^{2}\lambda_1\gamma_1^{-\delta}}{ \Gamma(2\kappa+\mu+1)\Gamma^{-1}(\mu)  \sigma_{3}^{2}\beta^{-(2\kappa+1)}}
\end{aligned}
\end{equation*}
with $\varrho$ defined in Theorem 1, $c_{3}=\Gamma(\mu+2\lambda)\Gamma^{-1}(\mu)$ and $c_4,c_5$ positive constants.
Then, using Theorem~1  of \cite{Stein:1993}, we obtain~\eqref{kaka}.
If \begin{equation}\label{kkk}
\left( \frac{\Gamma(2\kappa+\mu+1)}{\Gamma(\mu+1)} \right) \sigma_{3}^{2}\mu\beta^{-(2\kappa+1)}=\left(  \frac{\Gamma^2(\delta/2)\sin(\pi\delta/2)}{2^{1-\delta}\pi} \right) \sigma_{1}^{2}\lambda_1\gamma_1^{-\delta}
\end{equation}
and using Theorem~2 of \cite{Stein:1993},  we obtain (\ref{kaka2}).

\end{proof}

The implication of Point~1 is that under $P( {\cal C}_{\delta,\lambda_0,\gamma_0,\sigma_0^2})$,
prediction with $P({\cal C}_{\delta,\lambda_1,\gamma_1,\sigma_1^2})$
with an arbitrary $\gamma_1>0$
gives asymptotic prediction efficiency, if
$\delta \in (d/2,2)$,  $d=1,2,3$.
Moreover, if $\sigma_0^{2}\gamma_0^{-\delta}=\sigma_1^{2}\gamma_1^{-\delta}$
then  asymptotic prediction efficiency and asymptotically correct
estimates of error variance are achieved.
 By virtue of
Point~2, under $P({\cal M}_{\nu,\alpha,\sigma^2_2})$, prediction with
${\cal C}_{\delta,\lambda_1,\gamma_1,\sigma_0^2}$, with an arbitrary $\gamma_1>0$, gives
asymptotic prediction efficiency, if $\nu=\delta/2$, $\delta \in (d/2,2)$,  $d=1,2,3$.
Moreover if (\ref{ppp1}) is true
then  asymptotic
prediction efficiency and asymptotically correct estimates of error
variance are achieved.
Finally, Point 3  implies that
under $P({\cal \varphi}_{\mu,\kappa,\beta,\sigma^2_3})$, prediction with
$P({\cal C}_{\delta,\lambda_1,\gamma_1,\sigma_0^2})$, with an arbitrary $\gamma_1>0$, gives
asymptotic prediction efficiency, if $\kappa+1/2=\delta/2$,  $\delta \in (d/2,2)\cap [1,2)$,  $d=1,2,3$.
Moreover, if (\ref{kkk}) is true, then
asymptotic
prediction efficiency and asymptotically correct estimates of error
variance are achieved.

Theorem  \ref{kauf_3} is still valid interchanging the role  of the correct model with the wrong model.
For instance point 3 can be rewritten as follows.

\begin{theo}\label{kauf_4}
Let   $P({\cal C}_{\delta,\lambda_1,\gamma_1,\sigma_1^2})$,
$P(\varphi_{\mu,\kappa,\beta,\sigma^2_3})$ be two  Gaussian probability measures on $D\subset \R^d$,  $d=1,2,3$. Then, for all $\boldsymbol{s}_0\in D$ and
 under $P({\cal C}_{\delta,\lambda_1,\gamma_1,\sigma_1^2})$,   if $\kappa=\frac{\delta}{2}-\frac{1}{2}$, $\mu> d+\kappa+1/2$ and
   $\delta \in (d/2,2)\cap [1,2)$  as $n\to \infty$,
  \begin{equation}\label{kauf3_1111}
 U(\beta)=  \frac{\var_{\delta,\lambda_1,\gamma_1,\sigma_1^2}\bigl[\widehat{Z}_{n}(\mu,\kappa,\beta )-Z(\boldsymbol{s}_0)\bigr]}{\var_{\delta,\lambda_1,\gamma_1,\sigma_1^2}\bigl[\widehat{Z}_{n}(\delta,\lambda_1,\gamma)-Z(\boldsymbol{s}_0)\bigr]}{\,\longrightarrow\,}1,
        \end{equation}
           for any fixed $\beta>0$  and 
            if (\ref{kkk}) is true, then as $n\to \infty$
 \begin{equation}\label{kauf3_333} U_2=\frac{\var_{\mu,\kappa,\beta,\sigma^2_3}\bigl[\widehat{Z}_{n}(\mu,\kappa,\beta )-Z(\boldsymbol{s}_0)\bigr]}{\var_{\delta,\lambda_1,\gamma_1,\sigma_1^2}\bigl[\widehat{Z}_{n}(\mu,\kappa,\beta )-Z(\boldsymbol{s}_0)\bigr]}{\,\longrightarrow\,}1.
 \end{equation}
\end{theo}

 One remarkable
implication of  Theorem \ref{kauf_4} is that  when the true covariance belongs to the GC family,
asymptotic efficiency prediction and asymptotically correct estimation of
mean square error can be achieved, under suitable conditions, using a   compactly supported GW covariance model.

\section{Simulations and illustrations}

The main goals of this section are twofold: on the one hand, we compare
the finite sample behavior of the ML estimation of the microergodic
parameter of the GC model with the asymptotic distributions given in
Theorems~\ref{theo10} and~\ref{theo11}.
On the other hand, we  compare the finite sample behavior of MSE
prediction of a zero mean Gaussian process with GC covariance
model, using   a compatible  GW covariance model (Theorem~\ref{kauf_4}).

\smallskip

For the first goal
we have
considered  $4000$  points uniformly distributed over $[0,1]$
and then we randomnly select a sequence of
$n=250, 500, 1000$ points.
For each $n$ we simulate using Cholesky decomposition
and then we estimate with ML, $500$ realizations from a zero mean
Gaussian process with GC model. For the GC covariance model, ${\cal C}_{\delta,\lambda,\gamma_0,\sigma_0^2}$ we fix $\sigma_0^2=1$
and   in view of Theorem~\ref{theo11},
we fix   $\delta=0.75$ and $\lambda=1.5$.
Then we fix  $\gamma_0$ such that the practical range of the GC models is $0.3$, $0.6$ and $0.9$.
For a given correlation, with practical range $x$,  we mean that  the correlation is approximatively lower than $0.05$ when $r > x$.


   For each simulation, we consider
$\delta$ and $\lambda$  as known and fixed, and we estimate with ML
the variance and scale parameters, obtaining
$\hat{\sigma}^2_{i}$ and $\hat{\gamma}_i$, $i=1,\ldots,1000$.
In order  to estimate, we first  maximize the profile log-likelihood~\eqref{eq:prof} to
get $\hat{\gamma}_i$. Then, we obtain $\hat{\sigma}_i^2(\hat{\gamma}_i)=
\bz_i^{\prime}R(\hat{\gamma}_i)^{-1}\bz_i/n$, where $\bz_i$ is the data
vector of simulation $i$.
Optimization was carried out using the R \citep{R:2005} function \emph{optimize}
 where the
parametric space  was restricted to the interval $
[\varepsilon, 10\gamma_0]$ and $\varepsilon$ is slightly larger than
machine precision, about $10^{-15}$ here.

Using the asymptotic distributions stated in Theorems \ref{theo10} and \ref{theo11}, Table~\ref{tab21} compares
the sample quantiles of order $0.05, 0.25, 0.5, 0.75, 0.95$,
mean and variance of
$ \sqrt{n/2}\big(\widehat{\sigma}_i^2(x)\gamma_0^{\delta}/(\sigma_0^2x^{\delta} )-1   \big)$
for
$x=\widehat{\gamma}_i ,\gamma_0, 0.75\gamma_0, 1.25\gamma_0$ with the
associated theoretical values of the standard Gaussian distribution,
for  $n=500, 1000, 2000$.

%

As expected, the best approximation is achieved overall when using the
true scale parameter, i.e., $x=\gamma_0$.
In the case of
$x=\widehat{\gamma}_i$,  the sample
distribution converge to the  the asymptotic distribution given in
Theorem~\ref{theo11}   when increasing $n$, even if the convergence seems to be  slow.
Note that, for a fixed $n$, when increasing the practical range the convergence  to the standard Gaussian distribution is faster.
 In particular, for
$n=2000$ and practical range equal to $0.9$
the asymptotic distribution given in
Theorem~\ref{theo11} is a satisfactory approximation of the sample
distribution.
When using scale parameters  that are too small or too large with respect to
the true compact support ($x= 0.75\gamma_0, 1.25\gamma_0$), the convergence
to the asymptotic distribution given in Theorem~\ref{theo10} is very
slow.   These results are consistent with  \cite{Shaby:Kaufmann:2013}  and
 \cite{Bevilacqua_et_al:2016} and  when generating
 confidence intervals for the microergodic parameter
 we
 strongly recommend
jointly estimating variance and compact support and using the asymptotic
distribution given in Theorem~\ref{theo11}.


As for the second goal, using the results given in
Theorem~\ref{kauf_4}, we now  compare
asymptotic
prediction efficiency
 and asymptotically correct
estimation of prediction variance using ratios
$U(\beta)$ and $U_2$ defined in
~\eqref{kauf3_1111}
and  ~\eqref{kauf3_333}  respectively.
Specifically, we consider as true model  ${\cal C}_{\delta,\lambda_1,\gamma_1,\sigma_1^2}$    setting
$\sigma_1^2=1$, $\delta=1.2, 1.8$, $\lambda_1= 5$  and $\gamma_1$
such that the practical range is $0.3, 0.6, 0.9$.
As wrong model,
following the conditions  in
Theorem~\ref{kauf_4},
 we consider  $\varphi_{\mu,\kappa,\beta,\sigma^2_3}$
with $\sigma_3^2=1$, $\kappa=(\delta-1)/2$, $\mu=2+\kappa$
and  the ``equivalent'' compact support is obtained as:
\begin{equation*}
   \beta_1^*=\left[\gamma_1^{-\delta}\frac{\sigma_1^2\lambda_1 }{\sigma_3^2\mu}
   \frac{2^{\delta-1}sin(\pi \delta/2)\Gamma^2(\delta/2)\Gamma(\mu+1)}{\Gamma(2\kappa+\mu+1)\pi}
   \right]^{-1/(2\kappa+1)}.
\end{equation*}
For instance if $\delta=1.2$ and $\gamma_1$ is such that the practical range is equal to $0.3$
then $\beta_1^*=0.204$.
Figure  \ref{Fig:tre},  top left part, compares  the GW and GC  covariance model in this case.
The right part compares the GW and GC covariance model under the same setting but with $\delta=1.8$.
In Figure \ref{Fig:tre}, bottom part, we also show two realizations from
a Gaussian random process with
 the two compatible covariance functions shown in top left part.
 The two simulation are performed using cholesky decomposition and they share
 the same Gaussian simulation. It is apparent that the two realizations look very similar.

\begin{figure}
\begin{center}
\begin{tabular}{cc}
\includegraphics[width=0.4\textwidth]
{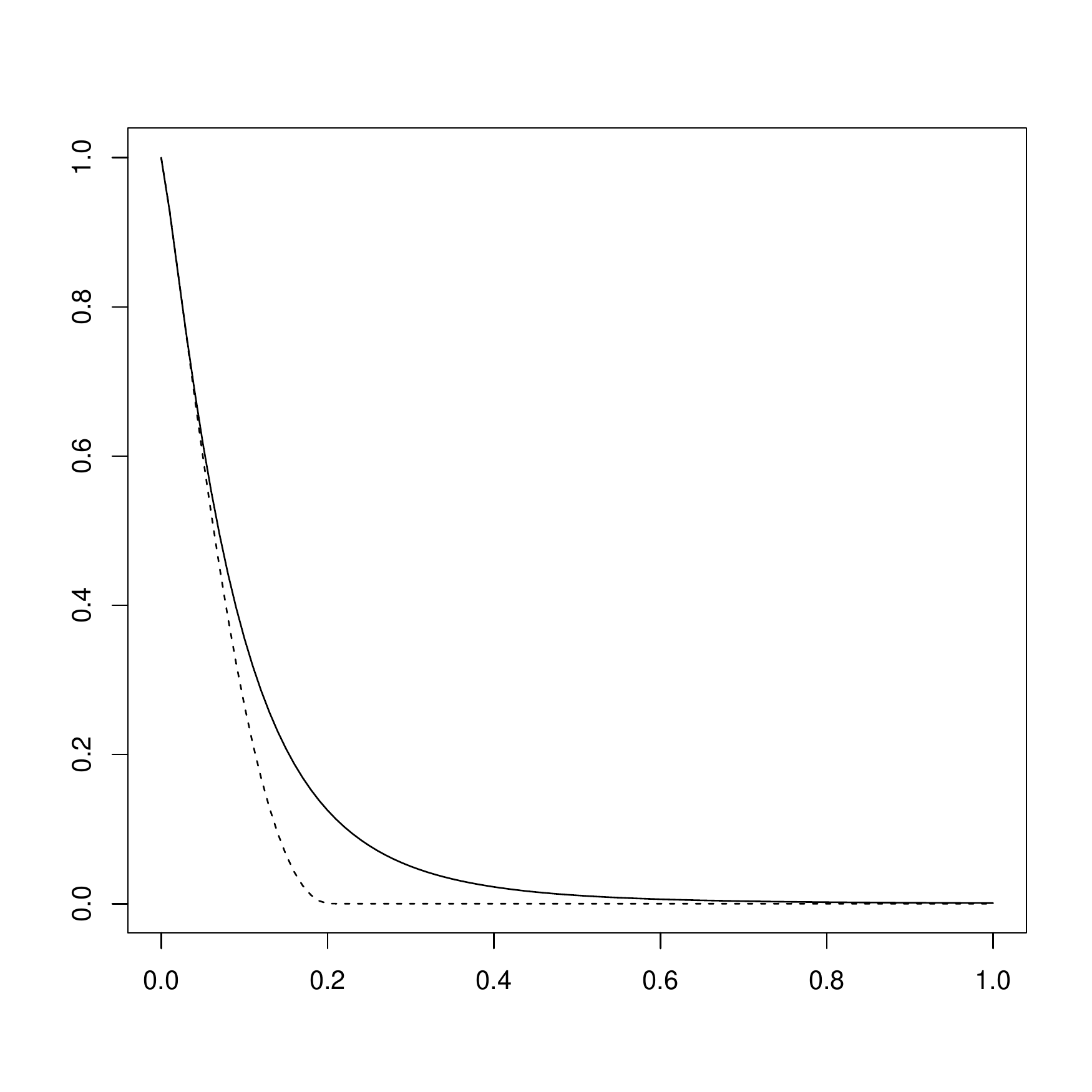}
&
\includegraphics[width=0.4\textwidth]
{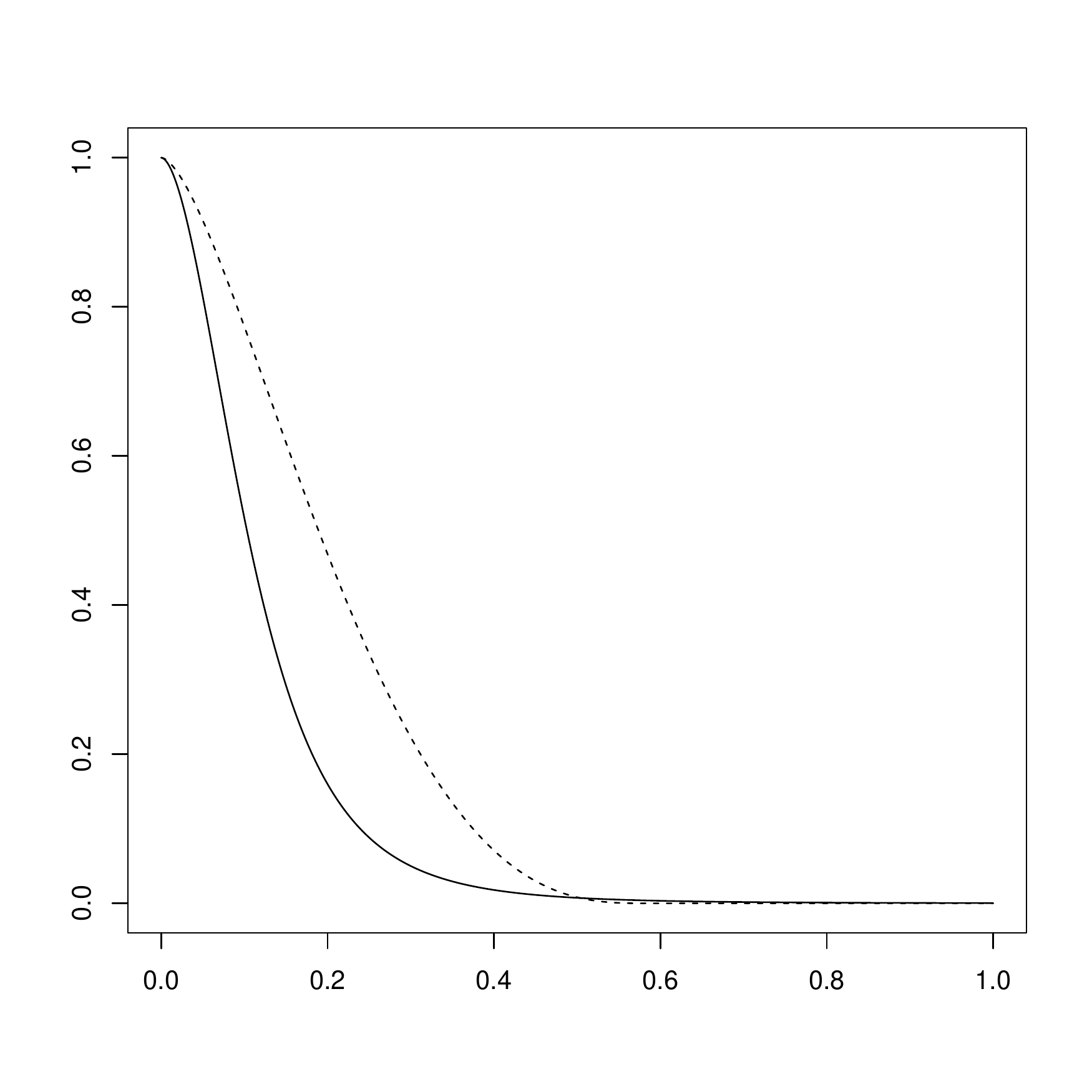}
\\
\includegraphics[width=0.4\textwidth]
{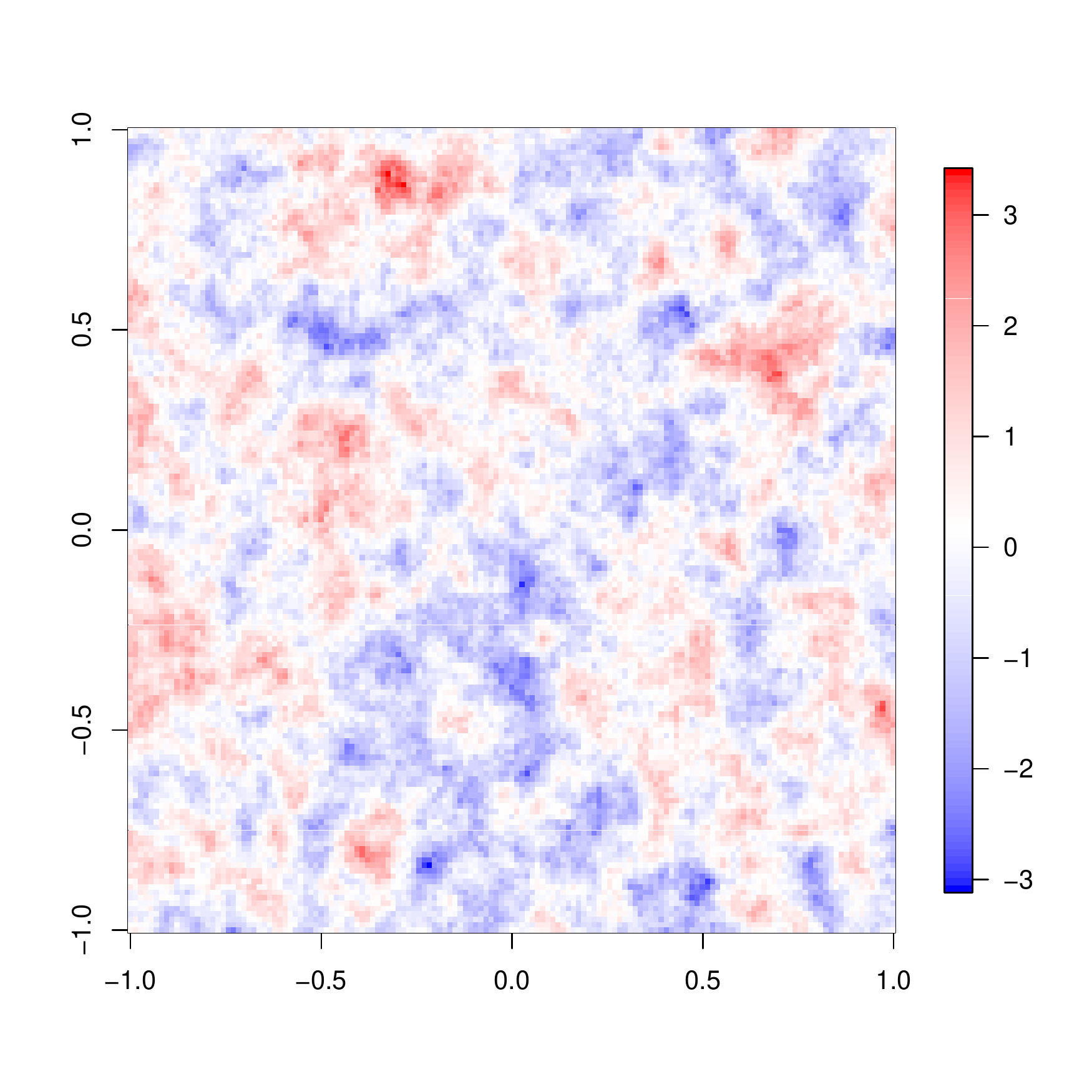}
&
\includegraphics[width=0.4\textwidth]
{RFgen.pdf}
\end{tabular}
\end{center}
\caption{Top left: a  ${\cal C}_{1.2,5,\gamma_1,1}$  model  (continous line)
and a       compatible   $\varphi_{2.1,0.1,\beta_1^*,1}$ model (dotted line).
Top right: A  ${\cal C}_{1.8,5,\gamma_1,1}$  model (continous line)
and a       compatible    $\varphi_{2.4,0.4,\beta_1^*,1}$ model (dotted line).
In both cases  $\gamma_1$ is chosen such that the practical range  is $0.3$ and $\beta^*_1$
is computed using the equivalence condition.
Bottom part: two realizations from two Gaussian random process
with  covariances  as shown in top left part (${\cal C}_{1.2,5,\gamma_1,1}$ on the left
and $\varphi_{2.1,0.1,\beta_1^*,1}$  on the right).
}\label{Fig:tre}
\end{figure}

Then
we randomly select $n_j=50, 100, 500, 1000$, $j=1,\ldots,100$
location sites without replacement from  $5000$  points uniformly distributed over $[0,1]^2$
and, for each $j$, we compute the ratio $\U_{1j}(x\beta_1^*)$, $x=1, 0.5, 2$ and the ratio
$\U_{2j}$, $j=1,\ldots,500$, using closed form expressions in
Equation~\eqref{mse_miss} and~\eqref{msetrue} when predicting the
location site $(0.26,0.48)^T$.
We consider $x=1, 0.5, 2$  in order
to investigate the effect of considering an arbitrary scale parameter on the convergence of ratio~\eqref{kauf3_1111}.

 Table~\ref{tab288} shows the
empirical means $\bar{\U}_{1}(x\beta_1^*)=\sum_{j=1}^{100}
\U_{1j}(x\beta_1^*)/100$ for $x=1, 0.5, 2$, and
$\bar{\U}_{2}=\sum_{j=1}^{100} \U_{2j}/100$ for $n_j$, $j=1,\ldots,100$.
Overall, the speed of convergence for both  $\bar{\U}_{1}(x\beta_1^*)$,  $x=1, 0.5, 2$
and $\bar{\U}_{2}$ is faster when increasing the dependence i.e. the practical range.
Additionally, as expected,  a too conservative choice of the arbitrary compact support ($0.5\beta_1^*$ in our simulations)
deteriorates the convergence of the ratio $\bar{\U}_{1}$.
These results are consistent with the results in  \cite{Bevilacqua_et_al:2016}.

It is interesting to note that
the speed  of convergence is clearly affected by the magnitude of $\delta$.  In particular for $\delta=1.8$
the convergence of both ratios is slower, in particular for  $\bar{\U}_{1}(x\beta_1^*)$,  $x=1, 0.5, 2$.
For instance, when the practical range is equal to 0.3, $n=1000$ is not sufficient to attain the convergence
for  $\bar{\U}_{1}(x\beta_1^*)$,  $x=1, 0.5, 2$.

\section{Concluding Remarks}\label{7}
  In this paper we studied estimation and prediction
of Gaussian processes with covariance models belonging to the GC family,
under fixed domain asymptotics.
Specifically, we first characterize the equivalence of two Gaussian
measures with CG models and then we establish strong
consistency and asymptotic Gaussianity of the ML estimator of the associated microergodic
parameter when considering both an arbitrary and an estimated scale parameter.
Simulation results show that for a finite sample, the choice
of an arbitrary scale parameter can result in a very poor
approximation of the asymptotic distribution.  These results are
consistent with those in \cite{Shaby:Kaufmann:2013}  in the MT
case and  \cite{Bevilacqua_et_al:2016} in the GW case.

We then  give  sufficient conditions for the equivalence of two
Gaussian measures with GW and GC model
 and two
Gaussian measures with MT and GC model
 and we
study the consequence of these results  on prediction under fixed domain asymptotics.

One remarkable  consequence of our results on optimal prediction  is that
the  mean square error prediction of
 a Gaussian process with a GC model
 can be achieved using a GW model under suitable conditions.

Then, under fixed domain asymptotics, a misspecified GW model can be used for optimal prediction when the
true covariance model is GC or MT  \citep{Bevilacqua_et_al:2016}.
GW  is an appealing  model from computational point of view since
the use of covariance functions with a compact support, leading to sparse matrices   (\cite{Furrer:2006},
\cite{Kaufman:Schervish:Nychka:2008}),
is a very accessible and scalable approach
and well established and implemented algorithms for sparse matrices can be
used when estimating the covariance parameters and/or predicting at
unknown locations (e.g.,  \cite{Furrer:Sain:2010}).
 An alternative strategy to produce sparse matrices is trough covariance tapering of the GC model but as outlined in \cite{Bevilacqua_et_al:2016},
 this kind of method is essentially an obsolete approach.

As highlighted  in Section 1, the parameter $\delta$ is crucial for the
differentiability at the origin and, as a consequence, for the degree of
differentiability of the associated sample paths.
Specifically,  for $\delta=2$, they are infinitely times differentiable
 and they are not differentiable for $\delta \in (0,2)$.
We do not  offer  results on  equivalence of Gaussian measures when $\delta=2$ and $0<\lambda< \infty$ for the GC family.
Nevertheless, it can be shown that
 ${\cal C}_{2,\lambda,  \sqrt{ \lambda\gamma/2},1}(r) \to  e^{-r^2/\gamma}$ as $\lambda \to \infty$.
This  result is consistent with the MT and GW cases
when considering the smoothness parameters going to infinity.
Specifically,
 ${\cal M}_{\nu,\sqrt{\alpha}/(2 \sqrt{\nu}),1}(r )\to e^{-r^2/\alpha}$ as $\nu\to \infty$
 and  $\varphi_{\mu,\kappa, g(\beta),1}(r) \to e^{-r^2/\beta}$
as $\kappa \to \infty$, where
 $g(\beta)= \sqrt{\beta}(\mu+2\kappa+1)\Gamma(\kappa+1/2)(2\Gamma(k+1))^{-1}$ \citep{che14}.

 Thus, rescaled versions of  GC, MT and GW
 converge to a squared exponential model
 when $\delta=2$ and $\lambda \to \infty$, $\nu\to \infty$ and $\kappa \to \infty$
 respectively.
Now, let $P({\cal G}_{\alpha_i,\sigma_i^2})$, $i=0, 1$ two zero mean Gaussian measures with squared exponential covariance function. In this case
$\widehat{{\cal G}}_{\alpha,\sigma^2}(z)=\sigma^2(\alpha/2)^{d/2}e^{-\alpha z^2/4}$ and  using (\ref{spectralfinite2}),
it can be shown that  the equivalence condition is given by  $\sigma_0^2=\sigma_1^2$, $ \alpha_0=\alpha_1$. Additionally,

$$  \underset{z \to \infty}{\lim}\frac{\widehat{{\cal G}}_{\alpha_1,\sigma_1^2}(z)}{\widehat{{\cal G}}_{\alpha_0,\sigma_0^2}(z)} =
\left\{%
\begin{array}{ll}
   0 , & \hbox{if \quad $\alpha_1>\alpha_0$}  \\
   +\infty, &  \hbox{if \quad $\alpha_1<\alpha_0$}   \\
    \sigma_1^2/   \sigma_0^2, & \hbox{if \quad $\alpha_1=\alpha_0$}  \\
\end{array}%
\right.
$$

and this implies that, under $P( {\cal G}_{\alpha_0,\sigma_0^2})$
and predicting with $P({\cal G}_{\alpha_1,\sigma_1^2})$,
 asymptotic prediction efficiency is achieved only when  $\alpha_0=\alpha_1$
 and
 asymptotically correct
estimates of error variance   are achieved under the  trivial condition $\sigma_0^2=\sigma_1^2$, $\alpha_0=\alpha_1$.



\section*{Acknowledgement}
Partial support was provided by FONDECYT grant 1160280, Chile
and  by Millennium
Science Initiative of the Ministry
of Economy, Development, and
Tourism, grant "Millenium
Nucleus Center for the
Discovery of Structures in
Complex Data"
for Moreno Bevilacqua.
The research work conducted by Tarik Faouzi was supported in part by
grant DIUBB 170308 3/I  Chile. Tarik Faouzi  thanks the support of project DIUBB 172409 GI/C at University of B{\'\i}o-B{\'\i}o.
\begin{table}
  \caption{Sample quantiles, mean  and variance of
 $\sqrt{n/2}(\widehat{\sigma}_i^2(x)\gamma_0^{\delta}/ (\sigma_0^2x^{\delta}) - 1)$, $i=1,\dots,1000$,
    for $x=\widehat{\gamma} ,\gamma_0, 1.25\gamma_0, 1.75\gamma_0$  when $\delta=0.75$, $\lambda=1.5$
     and $n=500, 1000, 2000$, compared with the associated theoretical
    values of the standard Gaussian distribution when $d=1$. Here  $\gamma_0$ is chosen such that the practical range is $0.3, 0.6, 0.9$.}\label{tab21}
\medskip\centering
\small{

\begin{tabular}{|c| c|  c| c| c| c| c| c| c| c| c|  }
	\hline
$Pr.range$  & $x$&	$n$  & 5$\%$ & 25$\%$ & 50$\%$ & 75$\%$ & 95$\%$ & Mean&Var\\
		\hline
	$0.3$&$\widehat{\gamma}$&500        &-1.806& -0.715&  0.052&  0.898 & 2.011& 0.115 &1.366\\
	&  &1000& -1.756&-0.724&-0.006&0.880&2.071&0.067&1.377\\
		&  &2000 &-1.749& -0.757 & 0.075&  0.751&  1.779& 0.022 &1.205\\

		\hline

	&$\gamma_0$&500       &-1.481 &-0.622 &-0.012&  0.746&  1.647& 0.055& 0.998\\
	&  &1000&-1.613&-0.745&-0.092&0.696&1.532&-0.040&1.018 \\
		&&2000 &-1.615 &-0.681& -0.015&  0.658 & 1.567 & -0.016 &1.005\\
		\hline

		&$1.25\gamma_0$&500       &-0.900& -0.025 & 0.640&  1.399&  2.348& 0.687& 1.086 \\
	& &1000 &-1.083&-0.178&0.453&1.281&2.202&0.518&1.068 \\
		&  &2000 &-1.152 &-0.216 & 0.483 &1.143  &2.088  &0.479& 1.038 \\
			\hline

			&$0.75\gamma_0$&500       &-1.809 &-0.982 &-0.383 & 0.351 & 1.247 &-0.320 &0.951 \\
	& &1000&-1.923&-1.076&-0.418&0.346&1.182&-0.375&0.990 \\
		&  &2000 &-1.884 &-0.962& -0.306&  0.363&  1.237& -0.314& 0.987 \\
	\hline \hline

	$0.6$&$\widehat{\gamma}$&500      &-1.685&-0.667&0.055&0.868&2.026&0.124&1.274 \\
	&  &1000&-1.678&-0.728&0.009&0.849&2.007&0.060&1.280\\
		&  &2000 &-1.692&-0.688&0.043&0.723&1.771&0.028&1.154\\

		\hline

	&$\gamma_0$&500       & -1.481&-0.622&-0.012&0.746&1.647&0.055&0.998\\
	&  &1000&-1.613&-0.745&-0.092&0.696&1.532&-0.040&1.018  \\
		&&2000 &-1.616&-0.681&-0.015&0.658&1.567&-0.016&1.005 \\
		\hline

		&$1.25\gamma_0$&500       &-1.104&-0.249&0.426&1.165&2.105&0.460&1.053 \\
	& &1000&-1.281&-0.388&0.252&1.058&1.949&0.308&1.049  \\
		&  &2000 &-1.341&-0.391&0.296&0.954&1.875&0.289&1.025\\
			\hline

			&$0.75\gamma_0$&500       &-1.693&-0.847&-0.247&0.487&1.391&-0.187&0.968\\
	& &1000&-1.810&-0.947&-0.293&0.472&1.312&-0.250&1.000 \\
		&  &2000 &-1.778&-0.858&-0.197&0.474&1.355&-0.200&0.994\\

				\hline \hline

$0.9$&$\widehat{\gamma}$&500       &-1.654&-0.651&0.076&0.848&1.979&0.130&1.232\\
	&  &1000&-1.688&-0.730&0.026&0.835&1.975&0.067&1.230  \\
		&  &2000 &-1.656&-0.712&0.051&0.708&1.751&0.032&1.131\\

		\hline

	&$\gamma_0$&500       &-1.481&-0.622&-0.012&0.746&1.647&0.055&0.998 \\
	&  &1000&   -1.613&-0.745&-0.092&0.696&1.532&-0.040&1.018\\
		&&2000 &-1.613&-0.681&-0.015&0.658&1.567&-0.016&1.005 \\
		\hline

		&$1.25\gamma_0$&500       & -1.193&-0.339&0.311&1.069&1.988&0.364&1.039 \\
	& &1000& -1.360&-0.475&0.170&0.973&1.845&0.222&1.041\\
		&  &2000 &-1.414&-0.465&0.219&0.882&1.797&0.212&1.020 \\
			\hline

			&$0.75\gamma_0$&500       & -1.644&-0.791&-0.189&0.547&1.452&-0.130&0.975 \\
	& &1000 &-1.761&-0.901&-0.245&0.532&1.364&-0.199&1.004 \\
		&  &2000 &-1.736&-0.814&-0.154&0.518&1.402&-0.154&0.997 \\
		\hline
		\multicolumn{3}{|c|}{$N (0,1)$} &-1.645 & -0.674 &  0 &  0.674 & 1.645&0&1\\
	\hline
\end{tabular}
}
\end{table}

\newpage
\begin{table}
\caption{$\bar{\U}_1(x)$, $x=0.5\beta_1^*, 2\beta_1^*,\beta_1^*$ and $\bar{\U}_2$ as defined in
(\ref{kauf3_1111}) and (\ref{kauf3_333}),
when considering a GC model  with  increasing practical range $(0.3, 0.6, 0.9)$,  smoothness parameter $\delta=1.2, 1.8$  and $n=50, 100, 500, 1000$.
Here $\beta_1^*$ is the  compact support parameter  of the GW model computed using the  equivalence condition.}\label{tab288}
\medskip\centering
\small{
\begin{tabular}{|c| c|c| c| c| c| c|  }
	\hline
$\delta$&   &$n$&$Pr.range= 0.3$& $Pr.range=0.6$ & $Pr.range=0.9$ \\
	\hline 
	$1.2$ & $\bar{\U}_{1}(0.5\beta_1^*)$  & 50  &1.164   &1.256 & 1.155\\
			&&   100& 1.237    &  1.172   &1.080  \\
			&&   500   &1.126 &  1.043   & 1.027 \\
			&&   1000   &1.068 & 1.029   &1.018  \\
				\hline
	&$\bar{\U}_{1}(2\beta_1^*)$  &50 &1.002  &1.035  &1.050  \\
		&&   100   &1.007 &1.052     & 1.047 \\
			&&   500   &1.055 &  1.036   & 1.024 \\
			&&   1000   & 1.047& 1.025   & 1.016 \\
		\hline
	&$\bar{\U}_{1}(\beta_1^*) $  & 50  &0.969  & 1.038& 1.052 \\
		&&   100   &0.999 &1.056    &1.048  \\
			&&   500   & 1.059& 1.037   &1.024  \\
			&&  1000   & 1.049&  1.026  &1.017 \\
					\hline
	&$\bar{\U}_{2}$  & 50  &0.973  & 0.987& 0.996 \\
			&&   100   & 0.979&  0.993  &0.998  \\
			&&   500   &0.994 &  0.998  & 0.999 \\
			&&   1000   & 0.996&  0.999  & 1.000 \\
	\hline
$1.8$& $\bar{\U}_{1}(0.5\beta_1^*)$  &50 & 2.696 & 2.055  &1.512  \\
&&  $ 100$   &2.533 & 1.575    & 1.232 \\
			&&   500   &1.423 &1.084    & 1.035 \\
			&&   1000   &1.215 &  1.040  &1.018  \\
				\hline
	&$\bar{\U}_{1}(2\beta_1^*) $  &50 & 2.812  &2.045 &1.508  \\
			&&   100   & 2.548&    1.566  &  1.231\\
			&&  500   &1.411 &  1.083  & 1.035 \\
			&&   1000   &1.209 & 1.039   &  1.018\\
		\hline
	&$\bar{\U}_{1}(\beta_1^*) $  &50 &2.810 &2.045&1.509 \\
			&&   100   &2.550 & 1.566   &1.231  \\
			&&   500   & 1.413& 1.083   &  1.035\\
			&&  1000   &1.210 & 1.039   & 1.018 \\
					\hline
	&$\bar{\U}_{2}$  &50  &0.944  &0.946 &  0.950\\
		&&  100   &0.958 &0.947    &0.973  \\
			&&  500  &0.960 &0.993    &0.999  \\
			&&   1000   &0.977 &0.998    & 1.000 \\
	\hline
\end{tabular}
}
\end{table}

\bibliographystyle{ECA_jasa}
\bibliography{mybib}

\end{document}